\documentclass{sig-alternate}

\usepackage{afterpage}
\usepackage{algpseudocode}
\usepackage{algorithm}
\usepackage{amsfonts}
\usepackage{flushend}
\usepackage[T1]{fontenc}
\usepackage{fp}
\usepackage{hyperref}
\usepackage{ifthen}
\usepackage[utf8]{inputenc}
\usepackage{listings}
\usepackage{microtype}
\usepackage{multicol}
\usepackage{multirow}
\usepackage{pgfplots}
\usepackage{stmaryrd}
\usepackage{url}
\usepackage{xcolor}
\usepackage{xfrac}
\usepackage{xkeyval}
\usepackage{xspace}

\newfont{\mycrnotice}{ptmr8t at 7pt}
\newfont{\myconfname}{ptmri8t at 7pt}
\let\crnotice\mycrnotice%
\let\confname\myconfname%

\begin{document}

\permission{Permission to make digital or hard copies of all or part of this
work for personal or classroom use is granted without fee provided that copies
are not made or distributed for profit or commercial advantage and that copies
bear this notice and the full citation on the first page.
Copyrights for components of this work owned by others than the author(s) must
be honored.
Abstracting with credit is permitted.
To copy otherwise, or republish, to post on servers or to redistribute to lists,
requires prior specific permission and/or a fee. Request permissions from
Permissions@acm.org.}
\conferenceinfo{SIGSOFT/FSE'14,}{November 16 - 22 2014, Hong Kong, China\\
{\mycrnotice{Copyright is held by the owner/author(s). Publication rights
licensed to ACM.}}}
\copyrightetc{ACM \the\acmcopyr}
\crdata{978-1-4503-3056-5/14/11\ ...\$15.00.\\
http://dx.doi.org/10.1145/2635868.2635912}

\title{Feedback Generation for Performance Problems in Introductory Programming Assignments
\thanks{The second and third author were supported by the Vienna Science and Technology Fund (WWTF) grant ICT12-059.}
}

\numberofauthors{3}
\author{
\alignauthor
Sumit Gulwani\\
\affaddr{Microsoft Research, USA}\\
\email{sumitg@microsoft.com}
\alignauthor
Ivan Radi\v{c}ek\\
\affaddr{TU Wien, Austria}\\
\email{radicek@forsyte.at}
\alignauthor
Florian Zuleger\\
\affaddr{TU Wien, Austria}\\
\email{zuleger@forsyte.at}
}

\maketitle

\def\istechreport{}
\newcommand\ignore[1]{}
\newcommand\delete[1]{}
\newcommand\modify[1]{{\color{blue}#1}}

\newtheorem{definition}{Definition}
\newtheorem{theorem}{Theorem}
\newtheorem{lemma}{Lemma}

\newcommand\equality{\ensuremath{E}}
\newcommand\equalities{\ensuremath{\mathcal{E}}}
\newcommand\traces{\ensuremath{\Gamma}}
\newcommand\criterion{\ensuremath{c}}
\newcommand\compare{\ensuremath{\delta}}
\newcommand\Full{\ensuremath{\mathit{full}}}
\newcommand\Partial{\ensuremath{\mathit{partial}}}
\newcommand\Id{\ensuremath{\mathit{Id}}}
\newcommand\prog{P}
\newcommand\progb{Q}
\newcommand\lang{L}
\newcommand\progvar{\mathit{Var}}
\newcommand\valdom{\mathit{Val}}
\newcommand\binvars{B}
\newcommand\var{v}
\newcommand\loc{\ell}
\newcommand\Loc{L}
\newcommand\val{\mathit{val}}
\newcommand\varmap{\pi}
\newcommand\progloc{\mathit{Loc}}
\newcommand\progtrace{\Gamma}
\newcommand\subseq{\sqsubseteq}
\newcommand\subseqq{\sqsubseteq}
\newcommand\vart[1]{\mathit{#1}}
\newcommand\dontcare{\mathbb{?}}
\newcommand\libset{L}
\newcommand\varset{V}
\newcommand\datadom{D}

\newcommand\default{\mathit{def}}
\newcommand\langl{$\mathcal{L}$\xspace}
\newcommand\dexpr{d}
\newcommand\nexpr{n}
\newcommand\bexpr{a}
\newcommand\expr{e}
\newcommand\varexpr{v}
\newcommand\stmt{s}
\newcommand\Stm{Stm}
\newcommand\BNFas{\;::=\;}
\newcommand\BNFalt{\;|\;}
\newcommand\opbin{op_{bin}}
\newcommand\opun{op_{un}}
\newcommand\keyw[1]{\texttt{#1}}
\newcommand\assign{:=}
\newcommand\compose{;}
\newcommand\observes{\keyw{observe}\xspace}
\newcommand\observesFun{\keyw{observeFun}\xspace}
\newcommand\nondet{\mathit{nd}}
\newcommand\eqfnc{E}
\newcommand\defeqfnc{E_=}
\newcommand\obsarg{o}
\newcommand\evaluate{\sigma}
\newcommand\state{\sigma}
\newcommand\trace{\gamma}
\newcommand\semmapsto{\mapsto}
\newcommand\semfun[2]{#1\!\left \llbracket #2 \right \rrbracket\!}
\newcommand\stmlab{l}
\newcommand\delim{;}
\newcommand\emptyseq{\epsilon}
\newcommand\covers{\keyw{cover}\xspace}
\newcommand\proglib{F}
\newcommand\libs{\keyw{f}\xspace}

\newcommand\invals{I}
\newcommand\inval{i}
\newcommand\algmain{Matches}
\newcommand\algtraceeq{Embed}
\newcommand\domset{D}
\newcommand\potential{G}

\newcommand\undef{\bot}
\newcommand\csharp{\textsc{C\#}\xspace}
\newcommand\clang{\textsc{C}\xspace}
\newcommand\pex{\textsc{Pex4Fun}\xspace}
\newcommand\tsc[1]{\textsc{#1}}
\renewcommand\paragraph[1]{{\bf #1}.}
\newcommand\incode[1]{\texttt{#1}}
\newcommand\ex[1]{{\bf\textsc #1}}

\newcommand\figlabel[1]{\label{fig:#1}}
\newcommand\figref[1]{Fig.~\ref{fig:#1}}
\newcommand\seclabel[1]{\label{sec:#1}}
\newcommand\secref[1]{\S\ref{sec:#1}}
\newcommand\tablabel[1]{\label{tab:#1}}
\newcommand\tabref[1]{Table~\ref{tab:#1}}
\newcommand\deflabel[1]{\label{def:#1}}
\newcommand\defref[1]{Definition~\ref{def:#1}}
\newcommand\thelabel[1]{\label{the:#1}}
\newcommand\theref[1]{Theorem~\ref{the:#1}}
\newcommand\lemlabel[1]{\label{lem:#1}}
\newcommand\lemref[1]{Lemma~\ref{lem:#1}}

\newcommand\smallspace[1]{\!\!\!\!\!#1\!\!\!\!\!}

\hypersetup{
  linkbordercolor=green
}

\algblockdefx[MyIf]{MyIf}{EndMyIf}%
    [1][Unknown]{{\bf if} #1{\bf :}}{}%
\algblockdefx[MyForAll]{MyForAll}{EndMyForAll}%
    [1][Unknown]{{\bf for all} #1{\bf :}}{}%
\algblockdefx[MyFunction]{MyFunction}{EndMyFunction}%
    [2][Unknown]{{\tsc #1}(#2){\bf :}}{}%
\newcommand\MyBreak{\text{{\bf break}}}

\algtext*{EndMyIf}%
\algtext*{EndMyForAll}%
\algtext*{EndMyFunction}%

\lstnewenvironment{mcode}%
  {\minipage{0.3\textwidth}%
    \lstset{%
      captionpos=b,
      language=Java,
      numbers=left,
      basicstyle=\fontsize{0.25cm}{0.1em}\ttfamily,
      columns=fullflexible,
      showstringspaces=false,
      numbersep=3pt,
      escapechar=\$,
      aboveskip=1pt,
      belowskip=1pt,
    }%
  }%
  {\endminipage}
\lstnewenvironment{mcodeH}%
  {\minipage{0.18\textwidth}%
    \lstset{%
      captionpos=b,
      language=C,
      numbers=left,
      basicstyle=\scriptsize\ttfamily,
      columns=fullflexible,
      showstringspaces=false,
      numbersep=3pt,
      escapechar=\$,
      aboveskip=1pt,
      belowskip=1pt,
    }%
  }%
  {\endminipage}
\lstnewenvironment{mcodeF}%
  {\minipage{0.18\textwidth}%
    \lstset{%
      captionpos=b,
      language=C,
      numbers=left,
      basicstyle=\scriptsize\ttfamily,
      columns=fullflexible,
      showstringspaces=false,
      numbersep=1pt,
      escapechar=\$,
      aboveskip=1pt,
      belowskip=1pt,
    }%
  }%
  {\endminipage}
\newcommand\chl[1]{\underline{#1}}

\newcommand\highlight[1]{{\color{blue}#1}}
\newcommand\todo[1]{
{\color{red}{\em#1}}
}

\ifdefined\istechreport
  \newcommand\techreport[1]{#1}
  \newcommand\techreporthide[1]{}
\else
  \newcommand\techreport[1]{}
  \newcommand\techreporthide[1]{#1}
\fi

\pgfplotscreateplotcyclelist{CycleList1}{%
  blue,every mark/.append style={fill=blue!80!black},mark=*\\%
  red,every mark/.append style={fill=red!80!black},mark=square*\\%
  brown!60!black,every mark/.append style={fill=brown!80!black},mark=diamond*\\%
}
\pgfplotscreateplotcyclelist{CycleList2}{%
  blue,every mark/.append style={fill=blue!80!black},mark=*\\%
  red,every mark/.append style={fill=red!80!black},mark=square*\\%
  brown!60!black,every mark/.append style={fill=brown!80!black},mark=diamond*\\%
  black,mark=asterisk\\%
  blue,every mark/.append style={fill=blue!80!black},mark=triangle*\\%
  red,mark=square\\%
  black,mark=o\\%
  blue,densely dashed,every mark/.append style={fill=blue!80!black},mark=pentagon*\\%
  red,mark=pentagon\\%
  black,mark=otimes\\%
  red,mark=diamond\\%
}

\newcommand\ExpLOCRatio{$\sfrac{L_S}{\overline{L_I}}$}
\newcommand\ExpLocS{$L_S$}
\newcommand\ExpLocI{$\overline{L_I}$}

\newcommand\printpercent[2]{\FPeval\result{round(#1*100/#2,1)}\result\%}

\makeatletter
\define@key{expkey}{name}{\gdef\MyExpName{#1}}
\define@key{expkey}{total}{\gdef\MyExpTotal{#1}}
\define@key{expkey}{correct}{\gdef\MyExpCorrect{#1}}
\define@key{expkey}{inefficient}{\gdef\MyExpInefficient{#1}}
\define@key{expkey}{algos}{\gdef\MyExpAlgos{#1}}
\define@key{expkey}{refines}{\gdef\MyExpRefines{#1}}
\define@key{expkey}{uncat}{\gdef\MyExpUncat{#1}}
\define@key{expkey}{inputs}{\gdef\MyExpInputs{#1}}
\define@key{expkey}{locratio}{\gdef\MyExpLOCRatio{#1}}
\define@key{expkey}{nondets}{\gdef\MyExpNonDets{#1}}
\define@key{expkey}{obsspec}{\gdef\MyExpObsSpec{#1}}
\define@key{expkey}{obsimpl}{\gdef\MyExpObsImpl{#1}}
\define@key{expkey}{spend}{\gdef\MyExpSpend{#1}}
\define@key{expkey}{mapall}{\gdef\MyExpMapAll{#1}}
\define@key{expkey}{mapfirst}{\gdef\MyExpMapFirst{#1}}
\define@key{expkey}{timeavg}{\gdef\MyExpTimeAvg{#1}}
\define@key{expkey}{timemax}{\gdef\MyExpTimeMax{#1}}
\makeatother

\newcommand\expkeyreset{
  \setkeys{expkey}{name=-,total=-,correct=-,inefficient=-,algos=-,uncat=-,
    inputs=-,nondets=-,obsspec=-,obsimpl=-,mapall=-,mapfirst=-,timeavg=-,
    timemax=-, spend=0, refines=-, locratio=-
  }
}
\newcommand\experiment[1]{%
  \expkeyreset%
  \setkeys{expkey}{#1}%
  \MyExpName 
  & \MyExpCorrect \; (\ifthenelse{\equal{\MyExpTotal}{-}}{--\%}{\printpercent{\MyExpCorrect}{\MyExpTotal}})
  & \MyExpInefficient \; (\printpercent{\MyExpInefficient}{\MyExpCorrect})
  & \MyExpAlgos & \MyExpRefines & \MyExpInputs
  & \MyExpNonDets
  & \MyExpLOCRatio &
  \MyExpObsSpec & \MyExpObsImpl 
  & \MyExpMapAll 
  & \MyExpTimeAvg & \MyExpTimeMax \\ \hline
}

\newenvironment{myitemize}{\begin{list}{$\bullet$}
{\setlength{\topsep}{1mm}\setlength{\itemsep}{0.25mm}
\setlength{\parsep}{0.1mm}
\setlength{\itemindent}{0mm}\setlength{\partopsep}{0mm}
\setlength{\labelwidth}{15mm}
\setlength{\leftmargin}{4mm}}}{\end{list}}

\begin{abstract}
Providing feedback on programming assignments manually 
is a tedious, error prone, and time-consuming task. In this paper, we
motivate and address the problem of generating feedback on performance
aspects in introductory programming assignments. 
We studied a large number of functionally correct student solutions to
introductory programming assignments and observed:
(1) There are different algorithmic strategies, with varying levels of
efficiency, for solving a given problem.
These different strategies merit different feedback. 
(2) The same algorithmic strategy can be implemented in countless different
ways, which are not relevant for reporting feedback on the student program.

We propose a light-weight programming language extension that allows a teacher
to define an algorithmic strategy by specifying certain key values that should
occur during the execution of an implementation.
We describe a dynamic analysis based approach to test whether a student's
program matches a teacher's specification.
Our experimental results illustrate the effectiveness of both our
specification language and our dynamic analysis. On one of our
benchmarks consisting of 2316 functionally correct implementations to 3
programming problems, we identified 16 strategies that we were able to describe
using our specification language
(in 95 minutes after inspecting 66, i.e., around 3\%, implementations).
Our dynamic analysis correctly matched each implementation with its corresponding
specification, thereby automatically producing the intended feedback.

\end{abstract}


\section{Introduction}\seclabel{Introduction}
Providing feedback on programming assignments is a very tedious, 
error-prone, and time-consuming task for a human teacher, even in a standard classroom setting.
With the rise of Massive Open Online Courses (MOOCs)~\cite{masters:MOOCs}, which have
a much larger number of students, this challenge is even more pressing. 
Hence, there is a need to introduce automation around this
task. Immediate feedback generation through automation can also enable new pedagogical benefits
such as allowing resubmission opportunity to students who submit imperfect solutions
and providing immediate diagnosis on class performance to a teacher who can then adapt her instruction accordingly~\cite{education:cacm}.

Recent research around automation of feedback
generation for programming problems
has focused on guiding students to functionally correct programs
either by providing counterexamples~\cite{nikolai:icse13} (generated
using test input generation tools) or
generating repairs~\cite{singh:pldi13}. 
However, non-functional aspects of a program, especially performance, are also important.
We studied several programming sessions of students who submitted
solutions to introductory \csharp programming problems on the \pex~\cite{PexForFun}
platform. In such a programming
session, a student submits a solution to a specified programming problem and receives a counterexample based
feedback upon submitting a functionally incorrect attempt (generated using Pex~\cite{pex:tap2008}). The student
may then inspect the counterexample and submit a revised attempt. This 
process is repeated until the student submits a functionally correct
attempt or gives up. We studied 24 different problems, and
observed that of the 3993 different programming sessions,
3048 led to functionally correct solutions. However,
unfortunately, on average around 60\% of these functionally correct solutions had
(different kinds of) performance problems. In this paper, we present a
methodology for semi-automatically generating appropriate performance
related feedback for such functionally correct solutions. 

From our study, we made two observations that form the basis of
our semi-automatic feedback generation methodology. 
(i) There are different {\em algorithmic strategies} with varying levels of
  efficiency, for solving a given problem.
Algorithmic strategies capture the global high-level insight of a solution to a
programming problem, while also defining key performance characteristics of the
solution. Different strategies merit different feedback.
(ii) The same algorithmic strategy can be implemented in countless
  different ways. These differences originate from local low-level implementation choices and are
not relevant for reporting feedback on the student program. 

In order to provide meaningful feedback to a student it is important to
identify what algorithmic strategy was employed by the student
program. A profiling based approach that measures
running time of a program or use of static bound
analysis techniques~\cite{speed:popl09,speed:pldi10}
is not sufficient for our purpose, because different
algorithmic strategies that necessitate different feedback may have
the same computational complexity. Also, a simple pattern
matching based approach is not sufficient because the same algorithmic
strategy can have syntactically different implementations. 

Our key insight is that the algorithmic strategy employed by a program
can be identified by observing the values computed during the execution of the
program. We allow the teacher to specify an algorithmic strategy by
simply annotating (at the source code level) certain key values computed by a sample program
(that implements the corresponding algorithm strategy) using a new
language construct, called {\em observe}. Fortunately, the number of different algorithmic strategies for
introductory programming problems is often small (at most 7 per
problem in our experiments). These can be easily enumerated by the teacher in an iterative
process by examining any student program that does not match
any existing algorithmic strategy; we refer to each such step
in this iterative process as an {\em inspection} step.

We propose a novel dynamic analysis that decides whether the student
program (also referred to as an {\em implementation}) matches an algorithm
strategy specified by the teacher in the form of an annotated program
(also referred to as a {\em specification}). Our dynamic analysis 
executes a student's implementation and the teacher's specification to
check whether the key values computed by the specification also occur 
in the corresponding traces generated from the implementation. 

We have implemented the proposed framework in \csharp and evaluated our
approach on 3 pre-existing programming problems on \pex (attempted by several
hundreds of students) and on 21 new problems that we hosted on \pex as part of
a programming course (attempted by 47 students in the course).
Experimental results show that:
(i) The manual teacher effort required to specify various algorithmic
strategies is a small fraction of the overall task that our system automates.
In particular, on our MOOC style benchmark of  2316 functionally correct
implementations to 3 pre-existing programming problems, we specified 16
strategies in 95 minutes after inspecting 66 implementations.
On our standard classroom style benchmark of 732 functionally correct
implementations to 21 programming problems, we specified 66
strategies in 266 minutes after inspecting 149
implementations.
(ii) Our methodology for specifying and matching algorithmic strategies is both
expressive and precise.
In particular, we were able to specify all 82 strategies using our
specification language and our dynamic analysis correctly matched each
implementation with the intended strategy.

This paper makes the following contributions:
\begin{myitemize}
\item We observe that there are different algorithmic strategies used in
  functionally correct attempts to introductory programming
  assignments; these strategies merit different performance related feedback.
\item We describe a new language construct, called {\em
    observe}, for specifying an algorithmic strategy (\secref{Language}).
\item We describe a dynamic analysis based approach to test whether a student's
  implementation matches the teacher's specification (\secref{Matching}).
\item Our experimental results illustrate the effectiveness of
  our specification language and dynamic analysis (\secref{Experiments}).
\end{myitemize}

\section{Overview}\seclabel{Overview}

\begin{figure*}[!htpb]
  \centering
  \scriptsize
  
\begin{tabular}{ccc}
  {\begin{mcode}
bool Puzzle(string s, string t) {
  if (s.Length != t.Length) return false;

  foreach (Char ch in s.ToCharArray()){
    if ($\chl{countChars(s, ch)}$
           != $\chl{countChars(t, ch)}$){
      return false;
    }}
  return true;}

int countChars(String s, Char c){
  int number = 0;

  foreach (Char ch in s.ToCharArray()){
    if ($\chl{ch}$ == c){
      number++;
    }}
  return number;}
  \end{mcode}} &
\begin{tabular}{c}
  {\begin{mcode}
bool Puzzle(string s, string t) {
  if (s.Length != t.Length)
    return false;
  else
    return s.All(c =>
      s.Where(c2 => $\chl{c2}$ == c).$\chl{Count()}$ ==
      t.Where(c2 => $\chl{c2}$ == c).$\chl{Count()}$
      );
}
  \end{mcode}} \\
(b) Counting/Library (\ex{C2}) \\ \hline
  {\begin{mcode}
bool Puzzle(string s, string t) {
  if(s.Length != t.Length) return false;
  foreach (var item in s) {
    if($\chl{s.Split}$(item).$\chl{Length}$
        != $\chl{t.Split}$(item).$\chl{Length}$)
    return false;
  }
  return true; }
  \end{mcode}}
\end{tabular} &
%
   {\begin{mcode}
int BinarySearch(List<char> xs, char y) {
  int low = 0, high = xs.Count;
  while (low < high) {
    int mid = (high - low) / 2 + low;
    if (y < xs[mid]) high = mid;
    else if (y > xs[mid]) low = mid + 1;
    else return mid;}
  return low;}

char[] Sort(string xs) {
  var res = new List<char>();
  foreach (var x in xs) {
    var pos = BinarySearch(res, x);
    res.Insert(pos, x);}
  return res.ToArray(); }

bool Puzzle(string s, string t) {
  return $\chl{String.Join("", Sort(s))}$
    == String.Join("", Sort(t)); }
   \end{mcode}} \\
(a) Counting/Manual (\ex{C1}) & (c) Counting/Split (\ex{C3}) &
(d) Sorting/Binary Insertion (\ex{S1}) \\ \hline
  {\begin{mcode}
bool Puzzle(string s, string t) {
  if (s.Length != t.Length) return false;
  char[] sa = s.ToCharArray();
  char[] ta = t.ToCharArray();
  for (int j=0; j < sa.Length; j++) {
    for (int i=0; i<sa.Length - 1;i++) {
      if (sa[i]<sa[i+1]){ char temp=sa[i];
        sa[i]=sa[i+1]; $\chl{sa[i+1]=temp}$;}
      if (ta[i]<ta[i+1]){ char temp=ta[i];
        ta[i] = ta[i+1]; ta[i+1] = temp;}
    }}
  for (int k = 0; k < sa.Length; k++) {
    if (sa[k] != ta[k]) return false; }
  return true; }
  \end{mcode}} &
\begin{tabular}{c}
  {\begin{mcode}
bool Puzzle(string s, string t) {
  var sa = s.ToCharArray();
  var ta = t.ToCharArray();
  $\chl{Array.Sort(sa)}$;
  Array.Sort(ta);
  return sa.SequenceEqual(ta);}
  \end{mcode}} \\
 (f) Sorting/Library (\ex{S3})\\ \hline
  {\begin{mcode}
bool Puzzle(string s, string t) {
  if (s.Length != t.Length) return false;
  foreach (char c in t.ToCharArray()) {
    int index = s.IndexOf(c);
    if (index < 0) return false;
    $\chl{s = s.Remove(index, 1)}$; }
  return true; }
  \end{mcode}}
\end{tabular} &
  {\begin{mcode}
bool Puzzle(string s, string t) {
  return IsPermutation(s, t);
}
bool IsPermutation(String s, string t) {
  if (s == t) return true;
  if (s.Length != t.Length) return false;
  int index = t.IndexOf(s[0]);
  if (index == -1) return false;

  s = s.Substring(1);
  $\chl{t = t.Remove(index, 1)}$;

  return IsPermutation(s, t);
}
   \end{mcode}} \\
(e) Sorting/Bubble (\ex{S2}) & (g) Removing/Library (\ex{R1})
& (h) Removing/Recursive (\ex{R2}) \\ \hline
   {\begin{mcode}
bool Puzzle(string s, string t) {
  char[] sc = s.ToCharArray();
  char[] tc = t.ToCharArray();
  Char c = '#';
  if(sc.Length!=tc.Length) return false;
  for(int i=0;i<sc.Length;i++) {
    c = sc[i];
    for(int j=0;j<tc.Length;j++) {
      if(tc[j]==c){
        $\chl{tc[j]='\#'}$;
        break;}
      if(j==tc.Length-1) {
        return false; }}}
  return true; }
   \end{mcode}} &
 {\begin{mcode}
Puzzle(string s, string t) {
  if ($\chl{nd1}$) { string tt = t; t = s; s = tt; }
  for (int i = 0; i < s.Length; ++i) {
    int cnt1 = 0, cnt2 = 0;
    for (int j = 0; j < s.Length; ++j) {
      if (s[j] == s[i]) {
        if ($\chl{nd2}$) $\chl{observe}$(s[j]);
        cnt1++;
      }}
    if (!$\chl{nd2}$) $\chl{observeFun}$(Split());
    $\chl{observe}$($\chl{nd2}$ ? cnt1 : cnt1 + 1);
    for (int j = 0; j < t.Length; ++j) {
      if (t[j] == s[i]) {
        if ($\chl{nd2}$) $\chl{observe}$(t[j]);
        cnt2++;
      }}
    if (!$\chl{nd2}$) $\chl{observeFun}$(Split());
    $\chl{observe}$($\chl{nd2}$ ? cnt2 : cnt2 + 1); }}
  \end{mcode}} &
 \begin{tabular}{c}
   {\begin{mcode}
Puzzle(string s, string t) {
  if ($\chl{nd1}$) s = s.ToUpperInvariant();
  char[] ca = s.ToCharArray();
  Array.Sort(ca);
  if ($\chl{nd2}$) Array.Reverse(ca);
  $\chl{observe}$(ca);
}
   \end{mcode}} \\
(k) Sorting Specification (\ex{SS}) \\ \hline
   {\begin{mcode}
Puzzle(string s, string t) {
  if ($\chl{nd1}$) {string tt = t; t = s; s = tt;}
  for (int i = 0; i < s.Length; ++i) {
    if (s.Substring(i) == t) return;
    int ni = $\chl{nd2}$ ? i : s.Length - i - 1;
    int k = $\chl{nd3}$ ? t.IndexOf(s[ni])
                     : t.LastIndexOf(s[ni]);
    t = t.Remove(k, 1);
    $\chl{observe}$(t, $\chl{CompareLetterString}$); }}
   \end{mcode}}
 \end{tabular} \\
(i) Removing/Manual (\ex{R3}) & (j) Counting Specification (\ex{CS}) &
(l) Removing Specification (\ex{RS}) \\ \hline
 {\begin{mcode}
bool Puzzle(string s, string t) {
  if(s.Length != t.Length) return false;
  Char[] taux = t.ToCharArray();
  for(int i = 0; i < s.Length; i++) {
    Char sc = s[i];
    Boolean exists = false;
    for(int j = 0; j < t.Length; j++) {
      if(sc == taux[j]) {
        exists = true; $\chl{taux[j] = ' '}$;
        break; }}
    if(exists == false) return false; }
  return true; }
    \end{mcode}} &
 {\begin{mcode}
bool CompareLetterString(string a,string b){
  var la = a.Where(x=>char.IsLetter(x));
  var lb = b.Where(x=>char.IsLetter(x));
  return la.SequenceEqual(lb);
}
    \end{mcode}} &
 {\begin{mcode}
bool Puzzle(string s, string t) {
  if (s.Length !=t.Length) return false;
  int [] cs=new int [256];
  int [] ct=new int [256];
  for(int i=0;i<s.Length;i++)
    $\chl{cs[(int) s[i]]++}$;
  for(int i=0;i<t.Length;++i)
    $\chl{ct[(int) t[i]]++}$;
  for (int i=0;i<256;i++)
    if(cs[i] != ct[i]) return false;
  return true;
}
 \end{mcode}} \\
(m) Removing/Manual 2 (\ex{R4}) & (n) Custom Data Equality (\ex{CDE}) &
(o) Efficient/Compare (\ex{E1}) \\  \hline
  {\begin{mcode}
bool Puzzle(string s, string t) {
  if (s.Length != t.Length)
    return false;
  char[] cs = s.ToCharArray();
  char[] ct = t.ToCharArray();
  int[] hash = new int[256];
  for (int i=0; i<255; ++i) {
    hash[i] = 0;
  }
  foreach (char ch in cs) {
    $\chl{hash[(int)ch]++}$; }
  foreach (char ch in ct) {
    $\chl{hash[(int)ch]-{-}}$; }
  for (int i=0; i<255; ++i) {
    if (hash[i] < 0)
      return false; }
  return true; }
    \end{mcode}} &
  {\begin{mcode}
void Puzzle(string s, string t) {
  if ($\chl{nd1}$){string tt = t; t = s; s = tt;}
  int[] cs = new int[256],ct = new int[256];
  $\chl{cover}$(ToCharArray());
  $\chl{cover}$(ToCharArray());
  $\chl{cover}$(255);
  for (int i = 0; i < s.Length; ++i) {
    cs[(int)s[i]]++;
    $\chl{observe}$(cs); }
  for (int i = 0; i < t.Length; ++i) {
    if ($\chl{nd2}$) { cs[(int)t[i]]--;
      $\chl{observe}$(cs);
    } else { ct[(int)t[i]]++;
      $\chl{observe}$(ct);
    }
  }
  $\chl{cover}$(255); }
    \end{mcode}} &
 {\begin{mcode}
bool Puzzle(string s, string t) {
  if (s.Length != t.Length) return false;
  string cp = t;
  for(int i=0; i<s.Length; i++) {
    char k = s[i]; bool found = false;
    for(int j=0; j<cp.Length; j++) {
      if (cp[j] == k) {
        if (j == 0) {
          $\chl{cp = (Char)0+cp.Substring(1)}$;}
        else if(j == cp.Length - 1) {
          $\chl{cp = cp.Substring(0, j)+(Char)0}$;}
        else {
          $\chl{cp = cp.Substring(0, j) +}$
            $\chl{(Char) 0 + cp.Substring(j + 1)}$;}
        found = true; break; }}
    if (!found) return false; }
  return true; }
   \end{mcode}} \\
(p) Efficient/Difference (\ex{E2}) & (q) Efficient Specification (\ex{ES}) &
(r) Removing/Separate computation (\ex{R5})

\end{tabular}

  \caption{Running example: Implementations and Specifications of Anagram assignment.}
  \figlabel{Motivation}
\end{figure*}

In this section we motivate our problem definition and various aspects
of our solution by means of various examples.

\subsection{Motivation}
\figref{Motivation} shows our running examples.
Programs (a)-(i) (\ex{IM}) show some sample implementations for the {\it anagram
  problem} (which involves testing whether the two input strings can be permuted to become equal)
  on the \pex platform. All 9 programs are examples of inefficient
implementations, because of their {\em quadratic asymptotic complexity}.
An efficient solution, for example, is to first collect (e.g. in an array) the
number of occurrences of each character in both strings and then compare them,
leading to {\em linear asymptotic complexity}.

\paragraph{Algorithmic strategies}
In implementations \ex{IM} we identify {\it three different algorithmic
  strategies}.
Implementations \ex{C1}-\ex{C3} iterate over
one of the input strings and for each character in that string count the
occurrences of that character in both strings
({\em counting strategy}).
Implementations \ex{S1}-\ex{S3} sort both
input strings and check if they are equal ({\em sorting strategy}).
Implementations \ex{R1}-\ex{R3} iterate over
one of the input strings and remove corresponding characters from the other
string ({\em removing strategy}).

\paragraph{Implementation details}
An algorithmic strategy can have several implementations.
In case of counting strategy: Implementation \ex{C1} calls manually implemented
method $\vart{countChar}$ to count the number of characters in a string (lines 5 and 6), while
implementation \ex{C2} uses a special \csharp construct (lines 6 and 7) and
implementation \ex{C3} uses the library function $\vart{Split}$ for that task (lines 4 and 5).
In case of the sorting strategy:
Implementation \ex{S1} employs binary insertion sort,
while implementation \ex{S2} employs bubble
sort and implementation \ex{S3} uses a library call (lines 4 and 5). We also
observe different ways of removing a character from a string
in implementations \ex{R1}-\ex{R3}.

\paragraph{Desired feedback}
Each of the three identified strategies requires separate feedback (independent
of the underlying implementation details),
to help a student understand and fix
the performance issues.
For the first strategy (implementations \ex{C1}-\ex{C3}),
a possible feedback might be:
{\em "Calculate the number of characters in each string in a preprocessing
  phase, instead of each iteration of the main loop"};
for the second strategy (\ex{S1}-\ex{S3}), it might be:
{\em "Instead of sorting input strings, compare the number of character
  occurrences in each string"};
and for the third strategy (\ex{R1}-\ex{R3}):
{\em "Use a more efficient data-structure to remove characters"}.

\subsection{Specifying Algorithmic Strategies}

\paragraph{Key values}
Our key insight is that different implementations that employ the
same algorithmic strategy generate the same {\it key
  values} during their execution on the same input.
For example, (the underlined expressions in) the implementations \ex{C1} and \ex{C2} both produce the
key value sequence $(a,b,a,2,b,a,a,2,a,b,a,1,b,a,a,1,a,b,a,2,b,a,a,2)$ on the input strings ``aba'' and ``baa''.

Our framework allows a teacher to describe an algorithmic strategy by simply annotating certain expressions
in a sample implementation using a special language statement \observes.
Our framework decides whether or not a student implementation matches a teacher
specification by comparing their execution traces on
common input(s).
We say that an implementation $\progb$ \emph{matches} a specification $\prog$, if
(1) the execution trace of $\prog$ is a subsequence of the execution trace of
$\progb$, and
(2) for every observed expression in $\prog$ there is an expression in $\progb$
that has generated the same values.
We call this matching criterion a {\em trace embedding}. The notion of
trace embedding establishes a fairly strong connection
between specification and implementation: basically, both programs produce the
same values at corresponding locations in the same order.
Our notion of trace embedding is an adaptation of the notion of a simulation
relation~\cite{sim:milner71} to dynamic analysis.

\paragraph{Non-deterministic choice}
Because of minor differences between implementations of the same strategy,
key-values can differ.
For example, implementation \ex{C3} uses a library function to obtain the
number of characters, while implementations \ex{C1} and \ex{C2} explicitly
count them by explicitly iterating over the string.
Moreover, counted values in \ex{C3} are incremented by one compared to those in
\ex{C1} and \ex{C2}.
\ex{C3} thus yields a different, but related, trace
$(\tsc{Split},3,\tsc{Split},3,\tsc{Split},2,\tsc{Split},2,\tsc{Split},3,\tsc{Split},3)$ on input strings "aba" and "baa".
To address variations in implementation details, we include a {\it non-deterministic} choice construct in our
specification language. The non-determinism is fixed before the execution; thus such a
choice is merely a syntactic sugar to succinctly represent {\em multiple similar specifications}
($n$ non-deterministic variables $=$ $2^n$ specifications).

\paragraph{Specifications}
\ex{CS}, \ex{SS}, and \ex{RS} denote the specifications for the counting
strategy (used in implementations \ex{C1}-\ex{C3}), sorting strategy (used in
\ex{S1}-\ex{S3}), and removing strategy (used in \ex{R1}-\ex{R3})
respectively. In \ex{CS}, the teacher observes the characters that are iterated
over (lines 7 and 14),  the results of counting the characters (lines 11 and
18), and use of  library function $\vart{Split}$ (lines 10 and 17).
Also the teacher uses {\em non-deterministic} Boolean variables: $\vart{nd1}$
(line 2) to choose the string over which the main loop iterates (as the input
strings are symmetric in the anagram problem); and $\vart{nd2}$ to choose
between manual and library function implementations (which also decides on
observed counted values on lines 18 and 11).
In \ex{SS} the teacher observes one of the input strings after sorting, and
non-deterministically allows that implementations convert input string to
upper-case ($\vart{nd1}$ on line 2), and sort the string in reverse order
($\vart{nd2}$ on line 5).
Notice that it is enough to observe only one sorted input, as in the case that
the input strings are anagrams, the sorted strings are the same.
In \ex{RS} the teacher observes the string with removed characters and
non-deterministcally chooses which string is iterated ($\vart{nd1}$ on line 2),
direction of the iteration ($\vart{nd2}$ on line 5) and the direction in which
the remove candidate is searched for ($\vart{nd3}$ on line 6).

\section{Specifications and Implementations}\seclabel{Language}

In this section we introduce an imperative programming language \langl that
supports standard constructs for writing implementations, and has some novel
constructs for writing specifications. 

\subsection{The Language \langl}

\begin{figure}[t]
\vspace{-0.3cm}
    \begin{eqnarray*}
      \mbox{Expression } \expr & \BNFas & \dexpr \BNFalt \var \BNFalt \var_1 ~ \opbin ~ \var_2 \BNFalt \opun ~ \var \BNFalt \var_1[\var_2] \\
      \mbox{Statement } \stmt & \BNFas & \var \assign \expr \BNFalt \var_1 [\var_2] \assign \expr \BNFalt \var \assign \libs(\var_1,\dots,\var_n) \\
      & \BNFalt & \stmt_0 \compose \stmt_1 \BNFalt \keyw{while} ~ \var ~ \keyw{do} ~ \stmt  \BNFalt \keyw{skip}\\
      & \BNFalt & \keyw{if} ~ \var ~ \keyw{then} ~ \stmt_0 ~ \keyw{else} ~ \stmt_1\\
      & \BNFalt & \observes(\var, [\equality])\\
      & \BNFalt & \observesFun(\libs[\var_1,\dots,\var_n], [\equality])
    \end{eqnarray*}
  \caption{The syntax of \langl language.}
  \figlabel{SyntaxL}
\end{figure}

The syntax of the language \langl is stated in \figref{SyntaxL}.
We discuss the features of the language below.

\paragraph{Expressions}
{\em A data value} $\dexpr$ is any value from some {\em data domain} set $\datadom$, which
contains all values in the language (e.g., in \csharp, all integers,
characters, arrays, hashsets, ...).
A {\em variable} $\var$ belongs to a (finite) set of variables $\progvar$.
An \emph{expression} is either a data value $\dexpr$, a variable $\var$, an
operator applied to variables $\var_1,\var_2$ or an array access
$\var_1 [\var_2]$.
Here, $\opbin$ represents a set of binary operators (e.g., $+,\cdot,=,<,\land$)
and $\opun$ a set of unary operators (e.g., $\lnot,|\cdot|$).
We point out that the syntax of \langl ensures that programs are in {\it three
  address code}:
operators can only be applied to variables, but not to arbitrary expressions.
The motivation for this choice is that three address code enables us to observe
any expression in the program by observing only variables.
We point out that any program can be (automatically) translated into three address code by
assigning each subexpression to a new variable.
For example, the statement $\var_1 \assign \var_2 + (a + b)$ can be translated
into three-address code as follows: $\var_3 \assign a + b; \var_1 \assign
\var_2 + \var_3$.
This enables us to observe the subexpression $a + b$ by observing $\var_3$.

\paragraph{Statements}
The statements of \langl allow to build simple imperative programs:
assignments (to variables and array elements), \keyw{skip} statement,
composition of statements, looping and branching constructs.
We also allow library function calls in \langl, denoted by
$\var \assign \libs(\var_1,\dots,\var_n)$, where $\libs \in \proglib$ is a library function
name, from a set of all library functions $\proglib$.
There are two special \observes constructs, which are only available to the
teacher (and not to the student).
We discuss the observe statements in \secref{Specification} below.
We assume that each statement $\stmt$ is associated with a {\it unique program
  location} $\loc$, and write $\loc: \stmt$.

\paragraph{Functions}
For space reasons we do not define functions here.
We could easily extend the language to (recursive) functions.
In fact we allow (recursive) functions in our implementation.

\paragraph{Semantics}
We assume some standard imperative semantics to execute programs written in the
language \langl (e.g., for \csharp we assume the usual semantics of \csharp).
The two \observes statements have the same semantic meaning of the \keyw{skip}
statement.

\paragraph{Computation domain}
We extend the data domain $\datadom$ by a special symbol $\dontcare$, which we
will use to represent {\em any} data value.
We define the {\em computation domain} $\valdom$ associated with our language \langl as
$\valdom = \datadom \cup (\proglib \times \datadom^*)$.
We assume the data domain $\datadom$ is equipped with some equality relation
$=_\datadom \subseteq \datadom \times \datadom$ (e.g., for \csharp we have
$(x,y)\in \ =_\datadom$ iff $a$ and $b$ are of the same type and comparison by
the {\tt equals} method returns {\tt true}).
We denote by $\equalities = 2^{\valdom \times \valdom}$ the set of all relations
over $\valdom$.
We define a {\em default equality relation}
$\equality_\default \in \equalities$ as follows:
We have $(x,y) \in \equality_\default$ iff $x = ?$ or $y = ?$ or $(x,y) \in =_\datadom$.
We have $((\libs,x_1,\dots,x_n),(\libs',y_1,\dots,y_n)) \in \equality_\default$
iff $\libs =\libs'$ and  $(x_i,y_i) \in \equality_\default$
for all $1 \le i \le n$.

\paragraph{Computation trace}
A \emph{(computation) trace} $\trace$ over some finite set of (programming) locations $\progloc$ 
is a finite sequence of location-value pairs $(\progloc \times \valdom)^*$.
We use the notation $\traces_\progloc$ to denote the set of all computation traces over $\progloc$.
Given some $\trace \in \traces_\progloc$ and $\progloc' \subseteq \progloc$,  we
denote by $\trace|_{\progloc'}$ the sequence that we obtain by deleting all
pairs $(\loc,\val)$ from $\trace$, where $\loc \not\in \progloc'$.

\subsection{Student Implementation}\seclabel{Implementation}

In the following we describe how a computation trace $\trace$ is generated for a
student implementation $\progb$ on a given input $\state$.
The computation trace is initialized to the empty sequence $\trace = \epsilon$.
Then the implementation is executed on $\state$ according to the semantics of
\langl.
During the execution we append location-value pairs to $\trace$ for every
assignment statement:
For $\loc: \var_1 \assign \expr$ or $\loc: \var_1[\var_2] \assign \expr$ we
append $(\loc,\evaluate(\var_1))$ to $\trace$ (we denote by $\evaluate(\var_1)$
the current value of $\var_1$).
We point out that we add the complete array $\evaluate(\var_1)$ to the
trace for an assignment to an array variable $\var_1$.
For a library function call $\loc : \var \assign \libs(\var_1,\dots,\var_n)$ we
append $(\loc,(\libs,\evaluate(\var),\evaluate(\var_1),\dots,\evaluate(\var_n)))$ to $\trace$.
We denote the resulting trace $\trace$ by $\semfun{}{\progb}(\state)$.
This construction of a computation trace can be achieved by {\em instrumenting the 
implementation} in an appropriate manner. 

\subsection{Teacher Specification}\seclabel{Specification}

The teacher uses \observes and \observesFun for specifying the key values she
wants to observe during the execution of the specification
and for defining an equality relation over computation domain.
As usual the rectangular brackets `[' and `]' enclose optional arguments.

In the following we describe how a computation trace $\trace$ is generated for a
specification $\prog$ on a given input $\state$.
The computation trace is initialized to the empty sequence $\trace = \epsilon$.
Then the specification is executed according to the semantics of \langl.
During the execution we append location-value pairs to $\trace$ only for
\observes and \observesFun statements:
For $\loc : \observes(\var, [\equality])$ we append
$(\loc,\evaluate(\var))$ to $\trace$ (we denote by $\evaluate(\var)$ the
current value of $\var$).
For $\loc : \observesFun(\libs[\var_1,\dots,\var_n],[\equality])$
we append $(\loc,(\libs,x_1,\dots,x_n))$ to $\trace$,
where $x_i = \evaluate(\var_i)$, if the $i^\text{th}$ argument to \libs has been
specified, and $x_i = ?$, if it has been left out.
We denote the resulting trace $\trace$ by $\semfun{}{\prog}(\state)$.

\paragraph{Custom data equality}
The possibility of specifying an equality relation $\equality \in \equalities$ at some location
$\loc$ is very useful for the teacher.
We point out that in practice the teacher has to specify $\equality$ by an
equality function $(\valdom \times \valdom) \to \{\keyw{true},\keyw{false}\}$.
The teacher can use $\equality$ to define the equality of {\em similar computation
  values}.
We show its usage on examples \ex{R3} and \ex{R4} (\figref{Motivation});
both examples implement the {\em removing strategy} (discussed in
\secref{Overview}) in almost identical ways --- the only difference is on lines
10 and 9, respectively, where implementations use different characters to
denote a character removed from a string: \incode{'\#'} and \incode{' '}.
In specification \ex{RS} the teacher uses the equality function
$\incode{CompareLetterString}$ (defined in \ex{CDE}) --- which
compares only letters of two strings --- to define value representations of
both implementations, regardless of used characters, as equal.

We call a function $\compare: \progloc \rightarrow \equalities$ a
\emph{comparison function}.
We define $\compare(\loc) = \equality$ for every
statement $\loc : \observes(\var,\equality)$ or
$\loc : \observesFun(\libs[\var_1,\dots,\var_n],\equality)$.
For statements, where $[\equality]$ has been left out,
we set the default value $\compare(\loc) = \equality_\default$.

\paragraph{Non-deterministic choice}
We assume that the teacher can use some finite set of {\it non-deterministic}
Boolean variables $\binvars = \{\nondet_1,\dots,\nondet_n \} \subseteq
\progvar$ (these are not available to the student).
Non-deterministic choice allows the teacher to specify variations in
implementations, as discussed in \secref{Overview}.
Non-deterministic variables are similar to the input variables, in the sense
that are assigned before program is executed.
We note that this results into $2^n$ different program behaviors for a given
input.

\section{Matching}\seclabel{Matching}

In this section, we define what it means for an implementation to {\it (partially) match} or {\it fully match} a specification and describe the corresponding matching algorithms.
The teacher has to determine for each specification which definition of matching has to be applied.
In case of partial matching we speak of \emph{inefficient specifications} and in case of full matching of \emph{efficient specifications}.

\subsection{Trace Embedding}\seclabel{Embedding}

We start out by discussing the problem of {\em Trace Embedding}  that we use as
a building block for the matching algorithms.

\paragraph{Subsequence}
We call
$\criterion \in \{\Partial,\Full\}$ a \emph{matching criterion}.
Let $\trace_1 = (\loc_1,\val_1) (\loc_2,\val_2) \cdots (\loc_n,\val_n)$ and
$\trace_2 = (\loc_1',\val_1') (\loc_2',\val_2') \cdots (\loc_m',\val_m')$ be
two computation traces over some set of locations $\progloc$, and let $\compare$ be some
comparison function (as defined in \secref{Specification}).
We say $\trace_1$ is a \emph{subsequence} of $\trace_2$ w.r.t.~to $\compare,\criterion$,
written $\trace_1 \subseqq_{\compare,\criterion} \trace_2$, if
there are indices $1 \leq k_1 < k_2 < \cdots < k_n \leq m$ such that for all
$1 \leq i \leq n$ we have $\loc_i = \loc_{k_{i}}'$ and $(\val_i, \val_{k_{i}}') \in \compare(\loc_i)$;
in case of $\criterion = \Full$ we additionally require that $\trace_1$ and $\trace_2|_{\{\loc_1,\ldots,\loc_n\}}$ have the same length.
We refer to $(\val_i, \val_{k_{i}}') \in \compare(\loc_i)$  as \emph{equality check}.
If $\compare(\loc_i) = \Id$ (the identity relation over $\valdom$) for all $1 \le i \le n$,
we obtain the usual definition of subsequence.

Since deciding subsequence, i.e., $\trace_1 \subseqq_{\compare,\criterion} \trace_2$,
is a central operation in this paper, we state complexity of
this decision problem.
It is easy to see that deciding subsequence requires only $O(m)$ equality
checks; basically one iteration over $\trace_2$ is sufficient.

\paragraph{Mapping Function}
Let $\progloc_1$ and $\progloc_2$ be two disjoint sets of locations.
We call an injective function $\varmap : \progloc_1 \to \progloc_2$ a
\emph{mapping function}.
We lift $\varmap$ to a function
$\varmap: \traces_{\progloc_1} \to \traces_{\progloc_2}$ by applying it to every location, i.e., we set
\begin{eqnarray*}
\varmap(\trace) & = & (\varmap(\loc_1),\val_1) (\varmap(\loc_2),\val_2) \cdots\\
\text{for } \trace & = &  (\loc_1,\val_1) (\loc_2,\val_2) \cdots \in \traces_{\progloc_1}.
\end{eqnarray*}

Given a comparison function $\compare$, a matching criterion $\criterion$,
and computation traces $\trace_1 \in \traces_{\progloc_1}$ and
$\trace_2 \in \traces_{\progloc_2}$
we say that $\trace_1$ {\em can be embedded in} $\trace_2$ by $\varmap$,
iff $\varmap(\trace_1) \subseqq_{\compare \circ \varmap^{-1},\criterion} \trace_2$, and write
$\trace_1 \subseqq^{\varmap}_{\compare,\criterion} \trace_2$.
We refer to $\varmap$ as {\em embedding witness}.

Executing a program on set of assignments $\invals$ gives rise to a set of traces, one for each assignment $\state \in \invals$.
We say that the set of traces $(\trace_{1,\state})_{\state \in \invals}$ can be
embedded in $(\trace_{2,\state})_{\state \in \invals}$
by $\varmap$ iff
$\trace_{1,\state} \subseqq^\varmap_{\compare,\criterion} \trace_{2,\state}$ for all $\state \in \invals$.

\begin{definition}[Trace Embedding]\deflabel{trace-embedding}
{\bf Trace \- Embedding} is the problem of deciding for given
sets of traces $(\trace_{1,\state})_{\state \in \invals}$ and $(\trace_{2,\state})_{\state \in \invals}$,
a comparison function $\compare$, and a matching criterion $\criterion$,
if there is a witness mapping function $\varmap$, such that
$\trace_{1,\state} \subseqq^{\varmap}_{\compare,\criterion} \trace_{2,\state}$ for all $\state \in \invals$.
\end{definition}

\paragraph{Complexity}
Clearly, Trace Embedding is in NP (assuming equality checks can be done in
polynomial time): we first guess the mapping function
$\varmap: \progloc_1 \rightarrow \progloc_2$ and then check
$\trace_{1,\state} \subseqq^{\varmap}_{\compare,\criterion} \trace_{2,\state}$
for all $\state \in \invals$ (which is cheap as discussed above).
However, it turns out that Trace Embedding is
NP-complete \emph{even} for a \emph{singleton} set $\invals$,
a \emph{singleton} computation domain $\valdom$, and the \emph{full matching criterion}.
\techreporthide{We present a detailed proof, by reduction from Permutation Pattern~\cite{journals/ipl/BoseBL98},
in our technical report~\cite{technical-report}.}

\techreport{
\begin{theorem}
\label{thm:complexity-trace-embedding}
  Trace Embedding is NP-complete (assuming equality checks can be done in
  polynomial time).
\end{theorem}

\begin{proof}
In order to show NP-hardness we reduce Permutation
Pattern~\cite{journals/ipl/BoseBL98} to Trace Embedding.
First, we formally define Permutation Pattern.
Let $n,k$ be positive integers with $k \le n$.
Let $\sigma$ be a permutation of $\{1,\cdots,n\}$ and let $\tau$ be a
permutation of $\{1,\cdots,k\}$.
We say $\tau$ \emph{occurs} in $\sigma$, if there is an injective function
$\pi: \{1,\cdots,k\} \rightarrow \{1,\cdots,n\}$ such that $\pi$ is monotone,
i.e., for all $1 \le r < s \le k$ we have $\pi(r) < \pi(s)$ and
$\pi(\tau(1))\cdots\pi(\tau(k))$ is a subsequence of
$\sigma(1)\sigma(2)\cdots\sigma(n)$.
\emph{Permutation Pattern} is the problem of deciding whether $\tau$ occurs in
$\sigma$.

We now give the reduction of Permutation Pattern to Trace Embedding.
We will construct two traces $\trace_1$ and $\trace_2$ over a singleton computation
domain $\valdom$, and over the sets of locations $\progloc_1 = \{1,\ldots,k\}$
and $\progloc_2 = \{1,\ldots,n\}$.
We set $\compare(i) = \Id$ (the identity function on $\valdom$) for every
$i \in \progloc_1$.
Because $\valdom$ is singleton, we can ignore values in the rest of the proof.
We set $\trace_1 = 12\cdots k \tau(1)\tau(2)\cdots\tau(k)$ and
$\trace_2 = 12\cdots n \sigma(1)\sigma(2)\cdots\sigma(n)$.
Because every $i \in \{1,\cdots,k\}$ occurs exactly twice in $\trace_1$ and
$\trace_2$, partial and full matching criteria are equivalent so we can ignore the difference.
We now show that $\tau$ occurs in $\sigma$ iff there is an injective function
$\varmap: \progloc_1 \rightarrow \progloc_2$ with
$\trace_1 \subseq^{\varmap} \trace_2$.
We establish this equivalence by two observations:
First, because every $i \in \{1,\cdots,k\}$ occurs exactly twice in $\trace_1$
and $\trace_2$ we have $12\cdots k \subseq^{\varmap} 12\cdots n$ and
$\tau(1)\tau(2)\cdots\tau(k) \subseq^{\varmap}
\sigma(1)\sigma(2)\cdots\sigma(n)$ iff $\trace_1 \subseq^{\varmap} \trace_2$.
Second, $12\cdots k \subseq^{\varmap} 12\cdots n$ iff $\varmap: \progloc_1
\rightarrow \progloc_2$ is monotone.
\end{proof}}

\paragraph{Algorithm}
\figref{AlgTrace} shows our algorithm, \tsc{\algtraceeq}, for the Trace Embedding problem.
A straightforward algorithmic solution for the trace embedding problem is to
simply test all possible mapping functions.
However, there is an {\it exponential number} of such mapping functions
w.r.t. to the cardinality of $\progloc_1$ and $\progloc_2$.
This exponential blowup seems unavoidable as the combinatorial search space is
responsible for the NP hardness. The {\it core element} of our algorithm is a
pre-analysis that narrows down the space of possible mapping functions
effectively.
We observe that if $\loc_2 = \varmap(\loc_1)$ and
$\trace_1 \subseqq^{\varmap}_{\compare,\criterion} \trace_2$, then there exists a trace
embedding restricted to locations $\loc_1$ and $\loc_2$, formally:
$\trace_1|_{\{\loc_1\}} \subseq^{\{\loc_1 \mapsto \loc_2\}}_{\compare,\criterion} \trace_2|_{\{\loc_2\}}$.
The algorithm uses this insight to create a (bipartite) graph
$\potential \subseteq \progloc_1 \times \progloc_2$ of potential mapping pairs
in lines 2-7.
A pair of locations $(\loc_1,\loc_2) \in \potential$ is a {\it potential
  mapping pair} iff there exists a trace embedding restricted to locations
$\loc_1$ and $\loc_2$, as described above.

The key idea in finding an embedding witness $\varmap$ is to construct a {\it
  maximum bipartite matching} in $\potential$.
A maximum bipartite matching has an edge connecting every program location from
$\progloc_1$ to a distinct location in $\progloc_2$ and thus gives rise to an
injective function $\varmap$.
We point out that such an injective function $\varmap$ does not need to be an
embedding witness, because, by observing only a single location pair at a time,
it ignores the order of locations.
Thus, for each maximum bipartite matching~\cite{Takeaki1997} $\varmap$ the algorithm checks (in lines 8-14)
if it is indeed an embedding witness.

\begin{figure}
  \begin{algorithmic}[1]
    \MyFunction[\algtraceeq]{$(\trace_{1,\state})_{\state \in \invals},(\trace_{2,\state})_{\state \in \invals},\progloc_1, \progloc_2,\compare,\criterion$}
      \State $\potential \gets \progloc_1 \times \progloc_2$
      \MyForAll[$\loc_1 \in \progloc_1,\loc_2 \in \progloc_2$]
        \MyForAll[$\state \in \invals$]
          \MyIf[$\trace_{1,\state}|_{\{\loc_1\}} \not\subseqq^{\{\loc_1 \mapsto \loc_2\}}_{\compare,\criterion} \trace_{2,\state}|_{\{\loc_2\}}$]
            \State $\potential \gets \potential \setminus \{(\loc_1,\loc_2)\}$
            \State \MyBreak
          \EndMyIf
         \EndMyForAll
      \EndMyForAll
      \MyForAll[$\pi \in \keyw{MaximumBipartiteMatching}(\potential)$]
        \State $\vart{found} \gets \keyw{true}$
        \MyForAll[$\state \in \invals$]
          \MyIf[$\trace_{1,\state} \not\subseqq^{\varmap}_{\compare,\criterion} \trace_{2,\state}$]
            \State $\vart{found} \gets \keyw{false}$
            \State \MyBreak
          \EndMyIf
        \EndMyForAll
        \MyIf[$\vart{found}=true$]
          \Return $\keyw{true}$
        \EndMyIf
      \EndMyForAll
      \State \Return $\keyw{false}$
    \EndMyFunction
  \end{algorithmic}
  \caption{Algorithm for Trace Embedding problem.}
  \figlabel{AlgTrace}
\end{figure}

The key strength of our algorithm is that it reduces the search space for
possible embedding witnesses $\varmap$.
The experimental evidence shows that this approach significantly reduces the number of possible
matchings and enables a very efficient algorithm in
practice, as discussed in \secref{Experiments}.

\subsection{Partial Matching} \seclabel{partial-matching}
We now define the notion of {\em partial matching} (also referred to simply as {\em matching}) which is used to check whether an implementation involves (at least)
those inefficiency issues that underlie a given inefficient specification.

\begin{definition}[Partial Matching]\deflabel{partial-matching}
Let $\prog$ be a specification with observed locations $\progloc_1$, let
$\compare$
be the comparison function
specified by $\prog$, and let $\progb$ be an implementation whose assignment
statements are labeled by $\progloc_2$.
Then implementation $\progb$ (partially) {\bf matches} specification $\prog$, on a set of
inputs $\invals$, if and only if there exists a mapping function
$\varmap : \progloc_1 \to \progloc_2$
and an assignment to the non-deterministic variables
$\state_\nondet$ such that $\trace_{1,\state}
\subseqq^{\varmap}_{\compare,\criterion} \trace_{2,\state}$, for all input
values $\state \in \invals$, where
$\trace_{1,\state} = \semfun{}{\prog}(\state \cup \state_\nondet)$,
$\trace_{2,\state} = \semfun{}{\progb}(\state)$ and $\criterion = \Partial$.
\end{definition}

\figref{AlgMain} describes an algorithm for testing if an implementation (partially) matches a given specification over a given set of input valuations $\invals$.
In lines 6-7, the implementation $\progb$ is executed on all input values
$\state \in \invals$.
In line 9, the algorithm iterates through all assignments $\state_\nondet$ to the
non-deterministic variables $\binvars_\prog$ of the specification $\prog$.
In lines 10-11, the specification $\prog$ is executed on all inputs $\state \in
\invals$.
With both sets of traces available, line 12 calls subroutine
\tsc{\algtraceeq} which returns \keyw{true} iff there exists a trace embedding witness.

\begin{figure}[t]
  \begin{algorithmic}[1]
    \MyFunction[\algmain]{$\text{Specification }\prog,\text{ Implementation }\progb,\text{ Inputs }\invals$}
      \State $\progloc_1$ = observed locations in $\prog$
      \State $\compare$ = comparison function specified by $\prog$
      \State $\criterion$ = matching criterion
      \State $\progloc_2$ = assignment locations of $\progb$
      \MyForAll[$\state \in \invals$]
        \State $\trace_{2,\state} \gets \semfun{}{\progb}(\state)$
      \EndMyForAll
      \State $\binvars_\prog = \text{ non-deterministic variables in }\prog$
      \MyForAll[$\text{assignments }\state_\nondet \text{ to } \binvars_\prog$]
        \MyForAll[$\state \in \invals$]
          \State $\trace_{1,\state} \gets \semfun{}{\prog}(\state \cup \state_\nondet)$
        \EndMyForAll
        \MyIf[$\tsc{\algtraceeq}((\trace_{1,\state})_{\state \in \invals},(\trace_{2,\state})_{\state \in \invals},\progloc_1,\progloc_2,\compare,\criterion)$]
          \State \Return{$\keyw{true}$}
        \EndMyIf
      \EndMyForAll
      \State \Return{$\keyw{false}$}
    \EndMyFunction
  \end{algorithmic}
  \caption{Matching algorithm.}
  \figlabel{AlgMain}
\end{figure}

\begin{figure}[t]
  \centering
  \footnotesize
  \begin{tabular}{c|cp{1pt}}
    \hspace{-0.17in}
    \begin{tabular}{cp{1pt}}
      \hspace{0.05in}
    {\begin{mcodeF}
Puzzle(s, t) {
  i = 0;
  n = |s|;
  while (i < n) {
    c = $\chl{s[i]}$;
    j = 0;
    cnt1 = 0;
    while (j < n) {
      c2 = $\chl{s[j]}$;
      if (c == c2) {
        cnt1 = cnt1 + 1; }
      j = j + 1; }
    j = 0;
    cnt2 = 0;
    while (j < n) {
      c2 = $\chl{t[j]}$;
      if (c == c2) {
        cnt2 = cnt2 + 1; }
      j = j + 1; }
    i = i + 1; }}
    \end{mcodeF}} &\vspace{-1.13in}(a)\\ \hline
  \hspace{0.05in}
    {\begin{mcodeF}
Puzzle(s, t) {
  i = 0;
  n = |s|;
  while (i < n) {
    c = $\chl{s[i]}$;
    ss = $\chl{\tsc{Split}(s,c)}$;
    cnt1 = |ss|;
    st = $\chl{\tsc{Split}(t,c)}$;
    cnt2 = |st|;
    i = i + 1; }}
    \end{mcodeF}} &\vspace{-0.57in}(b)
  \end{tabular} &
  \hspace{0.05in}
    {\begin{mcodeH}
Puzzle(s, t) {
  i = 0;
  n = |s|;
  while (i < n) {
    c = s[i];
    $\chl{observe}$(c);
    j = 0;
    cnt1 = 0;
    while (j < n) {
      c2 = s[j];
      if ($\chl{nd1}$)
        $\chl{observe}$(c2);
      j = j + 1; }
    j = 0;
    if (!$\chl{nd1}$)
      $\chl{observeFun}$($\tsc{Split}$());
    while (j < n) {
      c2 = t[j];
      if ($\chl{nd1}$)
        $\chl{observe}$(c2);
      j = j + 1; }
    if (!$\chl{nd1}$)
      $\chl{observeFun}$($\tsc{Split}$());
    i = i + 1; }}
    \end{mcodeH}} & \vspace{-1.73in}(c)%
  \end{tabular}
  \caption{Implementations (a), (b) and Spec. (c). 
  }
  \figlabel{FullExample}
\end{figure}

\paragraph{Example}
We now give an example that demonstrates our notion of programs and that
contains example applications of algorithms \tsc{\algtraceeq} and
\tsc{\algmain}.
In \figref{FullExample} we state two implementations, (a) and (b), and one
specification (c).
These programs represent simplified versions (transformed into three adress code) of \ex{R1} (after function
inlining), \ex{R3} and \ex{SC} (\figref{Motivation}).
Note, that every assignment and \observes statement is on its own line; we
denote line $i$ in program $x$ by by location $\loc_{x,i}$.
The argument $[\equality]$ 
has been left out for all locations in the specification,
thus we have $\compare(\loc) = \equality_\default$
for all specification locations $\loc$.

Algorithm \tsc{\algmain} runs all three programs on input values s = "aab" and t = "aba".
For program (a) we obtain the following computation trace:\\
$\trace_a=(\loc_{a,2},0)(\loc_{a,3},3)(\loc_{a,5},a)(\loc_{a,6},0)(\loc_{a,7},0)(\loc_{a,9},a)(\loc_{a,11},1)$
$(\loc_{a,12},1)(\loc_{a,9},a)(\loc_{a,11},2)(\loc_{a,12},2)(\loc_{a,9},b)(\loc_{a,12},3)(\loc_{a,13},0)\cdots$
Similarly, for program (b) we obtain:\\
$\trace_b=(\loc_{b,2},0)(\loc_{b,3},3)(\loc_{b,5},a)(\loc_{b,6},(\tsc{Split},aab,a))(\loc_{b,7},3)$\\
$(\loc_{b,8},(\tsc{Split},aba,a))(\loc_{b,9},3)(\loc_{b,10},1)(\loc_{b,5},a)\cdots$\\
For specification (c) we obtain two traces, depending on the choice for the non-deterministic variable $\vart{nd1}$:\\
$\trace_{c,\keyw{t}}=(\loc_{c,6},a)(\loc_{c,12},a)(\loc_{c,12},a)(\loc_{c,12},b)(\loc_{c,20},a)(\loc_{c,20},b)\cdots$\\
$\trace_{c,\keyw{f}}=(\loc_{c,6},a)(\loc_{c,16},(\tsc{Split},?,?))(\loc_{c,23},(\tsc{Split},?,?))\cdots$

Algorithm \tsc{\algmain} then calls \tsc{\algtraceeq} to check for trace embedding.
Algorithm \tsc{\algtraceeq} first constructs a potential graph $\potential$,
which contains an edge for two locations of the specification and the implementation
that show the same values.\\
For implementation (a), we obtain the following graph:
$\potential_a=\{(\loc_{c,6},\loc_{a,5}), (\loc_{c,6},\loc_{a,9}), (\loc_{c,6},\loc_{a,16}), (\loc_{c,12},\loc_{a,9}), (\loc_{c,20},\loc_{a,16})\}$.
Notice that $\loc_{c,6}$ shows the same values as the locations $\loc_{a,5},\loc_{a,9},\loc_{a,16}$ in the implementation (a).
However, there is only one maximal matching in $\potential_a$,
$\varmap_a=\{(\loc_{c,6},\loc_{a,5}), (\loc_{c,12},\loc_{a,9}), (\loc_{c,20},\loc_{a,16})\}$,
which is also an embedding witness; thus implementation (a) matches specification (c).\\
For implementation (b) and $\vart{nd1}=\keyw{true}$, we obtain the graph
$\potential_{b,\keyw{t}}=\{(\loc_{c,6},\loc_{b,5})\}$, from which we cannot construct
a maximal matching.
However, for $\vart{nd1}=\keyw{false}$, we obtain $\potential_{b,\keyw{f}}=\{(\loc_{c,6},\loc_{b,5}),(\loc_{c,16},\loc_{b,6}),(\loc_{c,23},\loc_{b,8})\}$,
which is also an embedding witness; thus implementation (b) matches specification (c).

\subsection{Full Matching}
Below we will define the notion of {\em full matching}, which is used to match
implementations against efficient specifications.
We will require that for every loop and every library function call in the implementation there is a corresponding loop and
library function call in the matching specification.
In order to do so, we need some helper definitions.

\paragraph{Observed loop iterations}
We extend the construction of the implementation trace (defined in \secref{Implementation}):
For each statement $\loc:~\keyw{while}~\var~\keyw{do}~\stmt$, we additionally
append element $(\loc,\bot)$ to the trace whenever the loop body $\stmt$ is entered.
We call $(\loc,\bot)$ a {\em loop iteration}.
Let $\varmap$ be a embedding witness s.t.,
$\trace_1 \subseqq^\varmap \trace_2$.
We say that $\varmap$ {\em observes all loop iterations} iff
between every two
loop iterations $(\loc,\bot)$ in $\trace_2$ there exists a
pair $(\loc',\val)$, such that $\exists\loc''.\varmap(\loc'')=\loc'$.
In other words, we require that between any two iterations of the same
loop, there exists some observed location $\loc'$.

\paragraph{Observed library function calls}
We say that $\varmap$ {\em observes all library function calls}
iff for every $(\loc,\libs(\val_1,\dots,\val_n))$ in $\trace_2$ there is a $\loc'$ such that
$\varmap(\loc')=\loc$.

\begin{definition}[Full Matching]\deflabel{full-matching}
Let $\prog$ be a specification with observed locations $\progloc_1$, let
$\compare$ be the comparison function
specified by $\prog$, and let $\progb$ be an implementation whose assignment
statements are labeled by $\progloc_2$.
Then implementation $\progb$ {\bf fully matches} specification $\prog$, on a set of
inputs $\invals$, if and only if there exists a mapping function
$\varmap : \progloc_1 \to \progloc_2$ and an assignment to the non-deterministic variables
$\state_\nondet$ such that
$\trace_{1,\state} \subseqq^{\varmap}_{\compare,\criterion} \trace_{2,\state}$, for all input
valuations $\state \in \invals$, where
$\trace_{1,\state} = \semfun{}{\prog}(\state \cup \state_\nondet)$,
$\trace_{2,\state} = \semfun{}{\progb}(\state)$, $\criterion = \Full$
and $\varmap$ observes all loop iterations and library function calls.
\end{definition}

We note that procedure \tsc{\algtraceeq} (\figref{AlgTrace}) can easily check at line 11 whether the current mapping $\varmap$ observes all loop iterations and library function calls.

It is tedious for a teacher to {\em exactly specify} all
possible loop iterations and library function calls used in different efficient implementations.
We add two additional constructs to the language \langl to simplify this specification task.

\paragraph{Cover statement}
We extend \langl by two {\em cover statements}:
$\loc : \covers(\libs[\var_1,\dots,\var_m],[\equality])$ and $\loc : \covers(\var)$.
The first statement is the same as the statement
$\loc : \observesFun($ $\libs[\var_1,\dots,\var_m],[\equality])$, except that we allow
the embedding witness $\varmap$ to not map $\loc$ to any location in the
implementation.
This enables the teacher to specify that function $\libs(\var_1,\dots,\var_m)$
{\em may appear} in the implementation.
The second statement allows $\varmap$ to map $\loc$ to a location $\loc'$ that
appears {\em at most} $\state(\var)$ times for each appearance of $\loc$, where
$\state(\var)$ is the current value of the specified variable $\var$.
Thus $\covers(\var)$ enables the teacher to {\em cover any loop} with up to
$\state(\var)$ iterations.

\paragraph{Example}
Now we present examples for efficient implementations (\ex{E1} and \ex{E2}) and
specification (\ex{ES}) for the Anagram problem (\figref{Motivation}).
The teacher observes computed values on lines 9, 12 and 14,
and uses a non-deterministic choice (on line 11)
to choose if implementations count the number of characters in each string,
or decrement one number from another.
Also the teacher allows {\em up to} two library function calls and two loops
with {\em at most} 255 iterations, defined by \covers statements on lines 4,5,6
and 17.

\section{Extensions}
In this section, we discuss useful extensions to the core material presented above.
These extensions are part of our implementation, but we discuss them separately to
make the presentation easier to follow.

\paragraph{One-to-many Mapping} According to definition of Trace Embedding, an embedding witness $\varmap$
maps one implementation location to a specification location, i.e.,
it constructs a {\em one-to-one} mapping.
However, it is possible that a student {\em splits} a computation of some value
over multiple locations.
For example, in the implementation stated in \ex{R5} (\figref{Motivation}), the student removes a
character from a string across three different locations (on lines 9, 11, 13
and 14), depending on the location of the removed character in the string.
This requires to map a \emph{single} location from the specification to
\emph{multiple} locations in the implementation!
For this reason, we extend the notion of trace embedding to {\em one-to-many} mappings
$\varmap: \progloc_1 \to 2^{\progloc_2}$ where $\varmap(\loc') \cap \varmap(\loc) = \emptyset$ for all $\loc \neq \loc'$.
It is easy to extend procedure \tsc{\algtraceeq}
(\figref{AlgTrace}) to this setting:
the potential graph $\potential$ is also helpful to enumerate every possible {\em one-to-many} mapping.
However, it is costly (and unnecessary) to search for arbitrary one-to-many
mappings.
We use heuristics to consider only a few one-to-many mappings.
For example, one of the heuristics in our implementation checks if the same
variable is assigned in different branches of an if-statement (e.g., in example
\ex{R5}, for all three locations there is an assignment to variable
$\vart{cp}$).

Although {\em many-to-many} mappings may seem more powerful, we
point out that the teacher can always write a specification that is more
succinct than the implementation of the student, i.e., the above described
one-to-many mappings provide enough expressivity to the teacher.

\paragraph{Non-deterministic behaviour}
\techreporthide{In our technical report~\cite{technical-report} we discuss another extension:
libraries with non-deterministic behaviour (e.g., the iteration order over a set).}
\techreport{Trace Embedding requires {\em equal values} in the {\em same order} in the
specification and implementation traces.
However, an implementation can use a library function with non-deterministic
behaviour, e.g., the values returned by a random generator or the iteration
order over a set data structure.
For such library functions we eliminate non-determinism by fixing one
particular behaviour, i.e, we fix the values returned by a random generator or
the iteration order over a set during program instrumentation.
These fixes do not impact functionally correct programs because they cannot
rely on some non-deterministic behaviour but allow us to apply our matching
techniques.}

\section{Implementation and Experiments}\seclabel{Experiments}

\begin{table*}[!htpb]
  \begin{center}
  \scriptsize
  \begin{tabular}{|c|c|c|c|c|c|c|c|c|c|c|c|c|} \hline
    {\bf Problem} 
    & {\bf Correct} 
    & {\bf Inefficient} 
    & \multirow{2}{*}{$N$}
    & \multirow{2}{*}{$S$}
    & \multirow{2}{*}{$I$}
    & \multirow{2}{*}{$\vart{ND}$}
    & \multirow{2}{*}{\ExpLOCRatio}
    & \multirow{2}{*}{$O_S$} & \multirow{2}{*}{$O_I$}
    & \multirow{2}{*}{$M$}
    & \multicolumn{2}{c|}{\bf Performance}\\ \cline{12-13}

    {\bf Name} 
    & {\bf Implement.} 
    & {\bf Implement.} 
    & 
    & 
    & 
    & 
    & 
    & & 
    & 
    & {\bf Avg.} & {\bf Max.} \\ \hline 
    \noalign{\smallskip} \hline

    \experiment{
      name=Anagram,
      total=766,
      correct=290,
      inefficient=261,
      algos=5,
      refines=25,
      inputs=1,
      locratio=1.41,
      nondets=3,
      obsspec=11,
      obsimpl=89,
      mapall=28357,
      mapfirst=28357,
      timeavg=0.42,
      timemax=7.67,
    }
    \experiment{
      name=IsSorted,
      total=1620,
      correct=1460,
      inefficient=139,
      algos=3,
      refines=23,
      inputs=2,
      locratio=1.45,
      nondets=2,
      obsspec=6,
      obsimpl=51,
      mapall=13,
      mapfirst=13,
      timeavg=0.33,
      timemax=1.31,
    }
    \experiment{
      name=Caesar,
      total=697,
      correct=566,
      inefficient=343,
      algos=5,
      refines=18,
      inputs=1,
      locratio=1.10,
      nondets=1,
      obsspec=7,
      obsimpl=39,
      mapall=172,
      mapfirst=172,
      timeavg=0.37,
      timemax=0.83,
    }

    \noalign{\smallskip} \hline

    \experiment{
      name=DoubleChar,
      total=47,
      correct=46,
      inefficient=31,
      algos=1,
      refines=5,
      inputs=1,
      locratio=0.72,
      nondets=0,
      obsspec=3,
      obsimpl=23,
      mapall=2,
      mapfirst=2,
      timeavg=0.31,
      timemax=0.42,
    }
    \experiment{
      name=LongestEqual,
      total=47,
      correct=37,
      inefficient=1,
      algos=1,
      refines=3,
      inputs=1,
      locratio=0.57,
      nondets=0,
      obsspec=1,
      obsimpl=35,
      mapall=2,
      mapfirst=2,
      timeavg=0.33,
      timemax=0.44,
    }
    \experiment{
      name=LongestWord,
      total=47,
      correct=39,
      inefficient=13,
      algos=2,
      refines=6,
      inputs=2,
      locratio=1.31,
      nondets=0,
      obsspec=7,
      obsimpl=46,
      mapall=15,
      mapfirst=15,
      timeavg=0.35,
      timemax=0.47,
    }
    \experiment{
      name=RunLength,
      total=44,
      correct=43,
      inefficient=32,
      algos=1,
      refines=6,
      inputs=1,
      locratio=0.90,
      nondets=0,
      obsspec=8,
      obsimpl=37,
      mapall=54,
      mapfirst=54,
      timeavg=0.33,
      timemax=0.44,
    }
    \experiment{
      name=Vigenere,
      total=44,
      correct=41,
      inefficient=32,
      algos=3,
      refines=5,
      inputs=1,
      locratio=0.64,
      nondets=0,
      obsspec=3,
      obsimpl=84,
      mapall=6,
      mapfirst=6,
      timeavg=0.34,
      timemax=0.50,
    }
    \experiment{
      name=BaseToBase,
      total=38,
      correct=15,
      inefficient=14,
      algos=2,
      refines=5,
      inputs=1,
      locratio=0.35,
      nondets=1,
      obsspec=3,
      obsimpl=64,
      mapall=13,
      mapfirst=13,
      timeavg=0.36,
      timemax=0.48,
    }
    \experiment{
      name=CatDog,
      total=47,
      correct=41,
      inefficient=8,
      algos=2,
      refines=18,
      inputs=1,
      locratio=2.02,
      nondets=1,
      obsspec=21,
      obsimpl=53,
      mapall=1629,
      mapfirst=1629,
      timeavg=0.36,
      timemax=0.58,
    }
    \experiment{
      name=MinimalDelete,
      total=38,
      correct=15,
      inefficient=8,
      algos=1,
      refines=8,
      inputs=2,
      locratio=2.21,
      nondets=3,
      obsspec=4,
      obsimpl=75,
      mapall=10,
      mapfirst=10,
      timeavg=0.86,
      timemax=4.36,
    }
    \experiment{
      name=CommonElement,
      total=45,
      correct=43,
      inefficient=32,
      algos=4,
      refines=14,
      inputs=2,
      locratio=0.97,
      nondets=1,
      obsspec=6,
      obsimpl=79,
      mapall=107,
      mapfirst=107,
      timeavg=0.36,
      timemax=0.53,
    }
    \experiment{
      name=Order3,
      total=46,
      correct=40,
      inefficient=30,
      algos=6,
      refines=12,
      inputs=1,
      locratio=1.45,
      nondets=2,
      obsspec=6,
      obsimpl=78,
      mapall=19,
      mapfirst=19,
      timeavg=0.40,
      timemax=0.59,
    }
    \experiment{
      name=2DSearch,
      total=44,
      correct=37,
      inefficient=36,
      algos=3,
      refines=7,
      inputs=1,
      locratio=1.09,
      nondets=1,
      obsspec=2,
      obsimpl=67,
      mapall=1,
      mapfirst=1,
      timeavg=0.34,
      timemax=0.45,
    }
    \experiment{
      name=TableAggSum,
      total=44,
      correct=11,
      inefficient=10,
      algos=1,
      refines=5,
      inputs=1,
      locratio=0.80,
      nondets=1,
      obsspec=3,
      obsimpl=144,
      mapall=1,
      mapfirst=1,
      timeavg=0.40,
      timemax=0.53,
    }
    \experiment{
      name=Intersection,
      total=44,
      correct=14,
      inefficient=12,
      algos=3,
      refines=7,
      inputs=2,
      locratio=0.89,
      nondets=1,
      obsspec=4,
      obsimpl=73,
      mapall=5,
      mapfirst=5,
      timeavg=0.37,
      timemax=0.56,
    }
    \experiment{
      name=ReverseList,
      total=40,
      correct=39,
      inefficient=0,
      algos=0,
      refines=3,
      inputs=1,
      locratio=0.35,
      nondets=0,
      obsspec=4,
      obsimpl=34,
      mapall=1,
      mapfirst=1,
      timeavg=0.34,
      timemax=0.44,
    }
    \experiment{
      name=SortingStrings,
      total=45,
      correct=41,
      inefficient=34,
      algos=5,
      refines=11,
      inputs=1,
      locratio=1.48,
      nondets=1,
      obsspec=13,
      obsimpl=110,
      mapall=866,
      mapfirst=866,
      timeavg=0.55,
      timemax=14.59,
    }
    \experiment{
      name=MinutesBetween,
      total=45,
      correct=45,
      inefficient=0,
      algos=0,
      refines=5,
      inputs=1,
      locratio=0.64,
      nondets=0,
      obsspec=8,
      obsimpl=101,
      mapall=1,
      mapfirst=1,
      timeavg=0.37,
      timemax=0.48,
    }
    \experiment{
      name=MaxSum,
      total=44,
      correct=42,
      inefficient=17,
      algos=2,
      refines=7,
      inputs=1,
      locratio=1.14,
      nondets=1,
      obsspec=2,
      obsimpl=51,
      mapall=3,
      mapfirst=3,
      timeavg=0.35,
      timemax=0.47,
    }
    \experiment{
      name=Median,
      total=47,
      correct=47,
      inefficient=47,
      algos=1,
      refines=1,
      inputs=1,
      locratio=0.39,
      nondets=0,
      obsspec=1,
      obsimpl=100,
      mapall=1,
      mapfirst=1,
      timeavg=0.34,
      timemax=0.44,
    }
    \experiment{
      name=DigitPermutation,
      total=36,
      correct=36,
      inefficient=1,
      algos=1,
      refines=3,
      inputs=1,
      locratio=0.26,
      nondets=0,
      obsspec=4,
      obsimpl=29,
      mapall=4,
      mapfirst=4,
      timeavg=0.32,
      timemax=0.44,
    }
    \experiment{
      name=Coins,
      total=41,
      correct=27,
      inefficient=14,
      algos=2,
      refines=6,
      inputs=1,
      locratio=1.65,
      nondets=1,
      obsspec=4,
      obsimpl=93,
      mapall=175,
      mapfirst=175,
      timeavg=2.41,
      timemax=15.44,
    }
    \experiment{
      name=Seq235,
      total=37,
      correct=33,
      inefficient=30,
      algos=4,
      refines=12,
      inputs=1,
      locratio=1.79,
      nondets=2,
      obsspec=3,
      obsimpl=232,
      mapall=3,
      mapfirst=3,
      timeavg=0.94,
      timemax=22.08,
    }

  \end{tabular}
  \end{center}
  \caption{List of all assignments with the experimental results.
    \ignore{}
  }
  \tablabel{ProblemList}
\end{table*}

\techreporthide{\begin{figure}
  \begin{tikzpicture}
    \begin{axis}[
        cycle list name=CycleList1,
        xlabel={\# of {\em inspection} steps},
        ylabel={\# of matched implementations},
        legend style={at={(0.6,0.8)},anchor={north west}},
        legend entries={Anagram,IsSorted,Caesar},
      ]
      \addplot table {data/match_anagram.txt};
      \addplot table {data/match_issorted.txt};
      \addplot table {data/match_caesar.txt};
    \end{axis}
  \end{tikzpicture}
  \vspace{-0.3cm}
  \figlabel{DiagramInspectionMOOC}
  \caption{The number of inspection steps required to completely specify the MOOC-style assignments.}
\end{figure}
}
\techreport{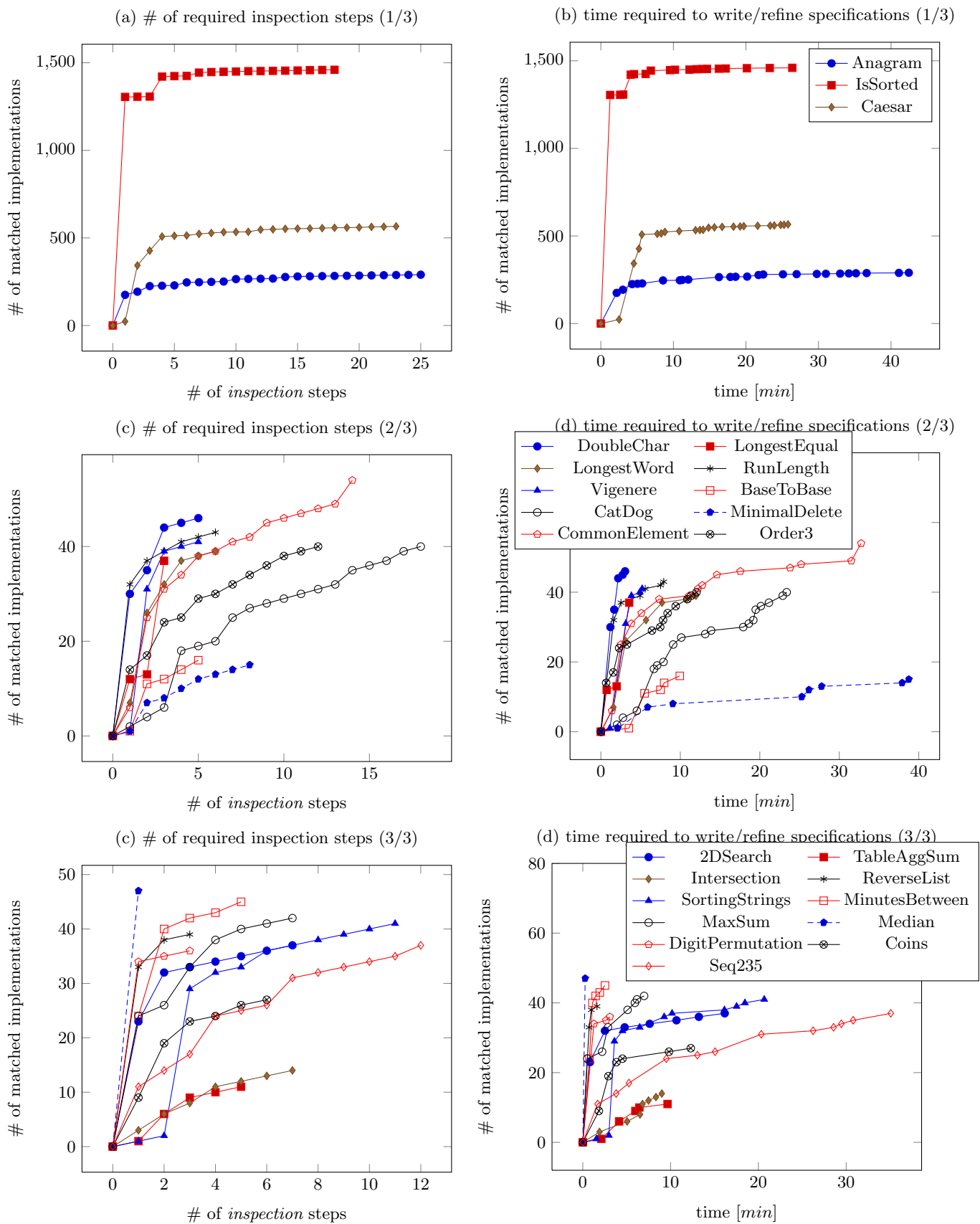
\begin{figure*}
  \begin{tabular}{cc}
  \begin{tikzpicture}
    \begin{axis}[
        cycle list name=CycleList1,
        title={(a) \# of required inspection steps (1/3)},
        xlabel={\# of {\em inspection} steps},
        ylabel={\# of matched implementations},
      ]
      \addplot table {data/match_anagram.txt};
      \addplot table {data/match_issorted.txt};
      \addplot table {data/match_caesar.txt};
    \end{axis}
  \end{tikzpicture}
  & \begin{tikzpicture}
    \begin{axis}[
        cycle list name=CycleList1,
        title={(b) time required to write/refine specifications (1/3)},
        xlabel={time $[min]$},
        ylabel={\# of matched implementations},
        legend entries={Anagram,IsSorted,Caesar},
      ]
      \addplot table {data/time_anagram.txt};
      \addplot table {data/time_issorted.txt};
      \addplot table {data/time_caesar.txt};
    \end{axis}
  \end{tikzpicture} \\
  \begin{tikzpicture}
    \begin{axis}[
        cycle list name=CycleList2,
        title={(c) \# of required inspection steps (2/3)},
        xlabel={\# of {\em inspection} steps},
        ylabel={\# of matched implementations},
      ]
      \addplot table {data/match_01DoubleChar.txt};
      \addplot table {data/match_02LongestEqual.txt};
      \addplot table {data/match_03LongestWord.txt};
      \addplot table {data/match_04RunLength.txt};
      \addplot table {data/match_05Vigenere.txt};
      \addplot table {data/match_06BaseToBase.txt};
      \addplot table {data/match_07CatDog.txt};
      \addplot table {data/match_08MinimalDelete.txt};
      \addplot table {data/match_09CommonElement.txt};
      \addplot table {data/match_10Order3.txt};
    \end{axis}
  \end{tikzpicture}
  & \begin{tikzpicture}
    \begin{axis}[
        cycle list name=CycleList2,
        title={(d) time required to write/refine specifications (2/3)},
        xlabel={time $[min]$},
        ylabel={\# of matched implementations},
        legend style={at={(-0.15,1.07)},anchor={north west}},
        legend columns = 2,
        ymax = 80,
        legend entries={DoubleChar,LongestEqual,LongestWord,RunLength,Vigenere,BaseToBase,CatDog,MinimalDelete,CommonElement,Order3},
      ]
      \addplot table {data/time_01DoubleChar.txt};
      \addplot table {data/time_02LongestEqual.txt};
      \addplot table {data/time_03LongestWord.txt};
      \addplot table {data/time_04RunLength.txt};
      \addplot table {data/time_05Vigenere.txt};
      \addplot table {data/time_06BaseToBase.txt};
      \addplot table {data/time_07CatDog.txt};
      \addplot table {data/time_08MinimalDelete.txt};
      \addplot table {data/time_09CommonElement.txt};
      \addplot table {data/time_10Order3.txt};
    \end{axis}
  \end{tikzpicture}  \\
  \begin{tikzpicture}
    \begin{axis}[
        cycle list name=CycleList2,
        title={(c) \# of required inspection steps (3/3)},
        xlabel={\# of {\em inspection} steps},
        ylabel={\# of matched implementations},
      ]
      \addplot table {data/match_11_2DSearch.txt};
      \addplot table {data/match_12TableAggSum.txt};
      \addplot table {data/match_13Intersection.txt};
      \addplot table {data/match_14ReverseList.txt};
      \addplot table {data/match_15SortingStrings.txt};
      \addplot table {data/match_16MinutesBetween.txt};
      \addplot table {data/match_17MaxSum.txt};
      \addplot table {data/match_18Median.txt};
      \addplot table {data/match_19DigitPermutation.txt};
      \addplot table {data/match_20Coins.txt};
      \addplot table {data/match_21Seq235.txt};
    \end{axis}
  \end{tikzpicture}
  & \begin{tikzpicture}
    \begin{axis}[
        cycle list name=CycleList2,
        title={(d) time required to write/refine specifications (3/3)},
        xlabel={time $[min]$},
        ylabel={\# of matched implementations},
        legend style={at={(0.2,1.07)},anchor={north west}},
        legend columns=2,
        ymax=80,
        legend entries={2DSearch,TableAggSum,Intersection,ReverseList,SortingStrings,MinutesBetween,MaxSum,Median,DigitPermutation,Coins,Seq235},
      ]
      \addplot table {data/time_11_2DSearch.txt};
      \addplot table {data/time_12TableAggSum.txt};
      \addplot table {data/time_13Intersection.txt};
      \addplot table {data/time_14ReverseList.txt};
      \addplot table {data/time_15SortingStrings.txt};
      \addplot table {data/time_16MinutesBetween.txt};
      \addplot table {data/time_17MaxSum.txt};
      \addplot table {data/time_18Median.txt};
      \addplot table {data/time_19DigitPermutation.txt};
      \addplot table {data/time_20Coins.txt};
      \addplot table {data/time_21Seq235.txt};
    \end{axis}
  \end{tikzpicture}
  \end{tabular}
  \figlabel{DiagramTech}
  \caption{The number of inspection steps and time required to completely specify assignments.}
\end{figure*}
}

We now describe our implementation and present an experimental evaluation of
our framework.
More details on our experiments can we found on the website~\cite{ArtifactsWeb}.

\subsection{Experimental Setup}

Our implementation of algorithm \tsc{\algmain} (\figref{AlgMain}) is in \csharp
and analyzes \csharp programs (i.e., implementations and specifications are in
\csharp).
We used Microsoft's Roslyn compiler framework~\cite{MSRoslyn}
for instrumenting every sub-expression to record value during program
execution.

\paragraph{Data}
We used 3 preexisting problems from \pex (as mentioned in
\secref{Introduction}):
(1) the {\em Anagram} problem, where students are asked to test if two strings
could be permuted to become equal,
(2) the {\em IsSorted} problem, where students are asked to test if the input
array is sorted, and
(3) the {\em Caesar} problem, where students are asked to apply Caesar cipher
to the input string.
We have chosen these 3 specific problems because they had a high number of
student attempts, diversity in algorithmic strategies and a problem was
explicitly stated (for many problems on \pex platform students have to guess the
problem from failing input-output examples).

We also created a new course on the \pex platform with 21 programming problems.
These problems were assigned as a homework to students in a second year
undergraduate course.
We created this course to understand performance related problems that {\it CS}
students make, as opposed to regular \pex users who might not have previous
programming experience.
We encouraged our students to write efficient implementations by giving more
points for performance efficiency than for mere functional correctness.
We omit the description of the problems here, but all
descriptions are available on the original course
page~\cite{PexCourse}.

\subsection{Methodology}\seclabel{TeachersWorkflow}
In the following we describe the methodology by which we envision the technique
in the paper to be used.

The teacher maintains a set of efficient and
inefficient specifications.
A new student implementation is checked against all available specifications.
If the implementation matches some specification, the associated feedback is
automatically provided to the student; otherwise the teacher is notified that there
is a new unmatched implementation. The teacher studies the implementation and
identifies one of the following reasons for its failure to match any existing
specification:
(i) The implementation uses a new strategy not seen before. In this case, the
teacher creates a new specification. 
(ii) The existing specification for the strategy used in the implementation is
too specific to capture the
implementation. In this case, the teacher refines that existing specification.
This overall process is repeated for each unmatched implementation.

\paragraph{New specification}
A teacher creates a new specification using the following steps: 
(i) Copy the code of the unmatched implementation. (ii) Annotate certain values and
function calls with observe statements.
(iii) Remove any unnecessary code (not needed in the specification) from the
implementation. (iv) Identify input values for the dynamic analysis for
matching. (v) Associate a feedback with the specification.

\paragraph{Specification refinement}
To refine a specification, the teacher identifies one of the following reasons
as to why an implementation did not match it: 
(i) The implementation differs in details specified in the specification; 
(ii) The specification observes more values than those that appear in the
implementation; 
(iii) The implementation uses different data representation.
In case (i) the teacher adds a new non-deterministic choice, and, if necessary,
observes new values or function calls; in case (ii) the teacher observes less
values; and in case (iii) the teacher creates or refines a custom data-equality.

\paragraph{Input values}
Our dynamic analysis approach requires the teacher to associate input
values with specifications. 
These input values should cause the corresponding 
implementations to exhibit their worst-case
behavior; otherwise an inefficient implementation might behave similar to an
efficient implementation and for this reason match the specification of the
efficient implementation.
This implies that {\em trivial} inputs should be avoided.
For example, two strings with unequal lengths constitute a trivial input for
the counting strategy since each of its three implementations \ex{C1}-\ex{C3} (\figref{Motivation})
then exit immediately. Similarly, providing a sorted input for the sorting
strategy is meaningless. We remark that it is easy for a teacher (who
understands the various strategies) to provide good input values.

\paragraph{Granularity of feedback}
We want to point out that the granularity of a feedback depends on the
teacher.
For example, in a programming problem where sorting the input value is an
inefficient strategy, the teacher might not want to distinguish between
different sorting algorithms, as they do not require a different feedback.
However, in a programming problem where students are asked to implement a
sorting algorithm it makes sense to provide a different feedback for different
sorting algorithms.

\subsection{Evaluation}

We report results on the 24 problems discussed above  in \tabref{ProblemList}.

\paragraph{Results from manual code study}
We first observe that a large number of students managed to write a
functionally correct implementation on most of the problems (column {\em
  Correct Implementations}).
This shows that \pex succeeds in guiding students towards a correct solution.

Our second observation is that for most problems a large fraction of
implementations is inefficient (column {\em Inefficient Implementations}),
especially for {\em Anagram} problem: 90\%.
This shows that although students manage to achieve functional correctness,
efficiency is still an issue (recall that in our homework the students were
explicitly asked and given extra points for efficiency).

We also observe that for all, except two, problems there is at least one
inefficient algorithmic strategy, and for most problems (62.5\%) there are
several inefficient algorithmic strategies (column $N$).
{\em These results highly motivate the need for a tool that can find
  inefficient implementations and also provide a meaningful feedback on how to
  fix the problem}.

\paragraph{Precision and Expressiveness}
For each programming assignment we used the above described methodology and
wrote a specification for each algorithmic strategy (both efficient and
inefficient).
We then {\em manually verified} that each specification matches all
implementations of the strategy, hence providing desired feedback for
implementations.
{\em This shows that our approach is {\em precise} and {\em expressive} enough to
  capture the algorithmic strategy, while ignoring low level implementation
  details.}

\paragraph{Teacher Effort}
To provide manual feedback to students the teacher would have to go through
every implementation and look at its performance characteristics.
In our approach the teacher has to take a look only at a few representative
implementations.
In column $S$ we report the total number of inspection steps that we required
to fully specify one programming
problem, i.e., the number of implementations that the teacher would had to go
through to provide feedback on all implementations.
For the 3 pre-existing problems {\em the teacher would
  only have to go through 66 out of 2316 (or around 3\%) implementations to
  provide full feedback}.
\techreporthide{\figref{DiagramInspectionMOOC} shows the number of matched
implementations with each inspection step (by inspecting the first unmatched
implementation in the random order).
For space reasons the diagram shows only 3 pre-existing assignments;
our technical report~\cite{technical-report} shows all 24 assignments as
well as the time it took us for each inspection step.}
\techreport{\figref{DiagramTech} shows the number of matched implementations
  with each inspection step, as well the time it took us to create/refine all
  specifications (we measured the time it takes from seeing an unmatched
  implementation, until writing/refining a matching specification for it).}

In column \ExpLOCRatio we report the largest ratio of specification and
average matched implementation in terms of lines of code.
We observe that in half of the cases the largest specification is about the same size
or smaller than the average matched implementation.
Furthermore, the number of the input values that need to be provided by the
teacher is 1-2 across all problems (column $I$).
In all but one problem ({\em IsSorted}) one set of input values is used for all
specifications.
Also, in about one third of the specifications there was no need for
non-deterministic variables, and the largest number used in one specification is
3 (column $\vart{ND}$).
{\em Overall, our semi-automatic approach requires considerably less teacher
  effort than providing manual feedback}.

\paragraph{Performance}
We plan to integrate our framework in a MOOC platform, so performance, as for
most web applications, is critical.
Our implementation consists of two parts.
The first part is the execution of the implementation and the specification
(usually small programs) on relatively small inputs and obtaining execution
traces, which is, in most cases, neglectable in terms of performance.
The second part is the \tsc{\algtraceeq} algorithm.
As discussed in \secref{Embedding} the challenge consists in finding an
embedding witness $\varmap$.
With $O_S$ observed variables in the specification and $O_I$ observed variables
in the implementation, there are $\frac{O_I!}{(O_I-O_S)!}$ possible injective
mapping functions.
E.g., for the {\em SortingStrings} problem that gives $\approx 10^{26}$
possible mapping functions ($O_I=110,O_S=13$).
However, our algorithm reduces this huge search space by constructing a
bipartite graph $\potential$ of potential mappings pairs.
In $M$ we report the number of
mapping functions that our tool had to explore.
E.g., for {\em SortingStrings} only 866 different mapping functions had to be
explored.
For all values ($O_S$, $O_I$ and $M$) we report the maximal number across
all specifications.
In the last column we state the {\em total execution time} required to decide
if one implementation matches the specification (average and maximal).
Note that this time includes execution of both programs, exploration of all
assignments to non-deterministic Boolean variables and finding an embedding witness
$\varmap$.
Our tool runs, in most cases, under half a second per implementation.
{\em These results show that our tool is fast enough to be used in an
  interactive teaching environment.}

\subsection{Threats to Validity}
\paragraph{Unsoundness}
Our method is unsound in general since it uses a dynamic analysis that explores
only a few possible inputs.
However, we did not observe any unsoundness in our large scale experiments.
If one desires provable soundness, an embedding witness could be used as a
guess for a simulation relation that can then be formally verified by other
techniques. Otherwise, a student who suspects an incorrect feedback can always
bring it to the attention of the teacher.
\ignore{Old}

\paragraph{Program size}
We evaluated our approach on introductory programming assignments.
Although questions about applicability to larger programs might be raised, our
goal was not to analyze arbitrary programs, but rather to develop a framework
to help teachers who teach introductory programming with providing performance
feedback --- currently a manual, error-prone and time-consuming task.

\paragraph{Difficulty of the specification language}
Although we did not perform case study with third-party instructors, we report
our experiences with using the proposed language.
We would also like to point out that writing specifications is a one-time
investment, which could be performed by an experienced personnel.

\section{Related Work}
\seclabel{Related}

\subsection{Automated Feedback}
There has been a lot of work in the area of generating
automated feedback for programming assignments. This work can be
classified along three dimensions: (a) aspects on which the feedback is
provided such as functional correctness, performance characteristics
or modularity (b) nature of the feedback such as counterexamples, bug
localization or repair suggestions, and (c) whether static or dynamic
analysis is used.

Ihantola et.al.~\cite{Ihantola} present a
survey of various systems developed for automated grading of programming
assignments. The majority of these efforts have focussed on checking
for functional correctness. This is often done by examining
the behavior of a program on a set of test inputs.
These test inputs can be manually written or automatically
generated~\cite{nikolai:icse13}. There has only been little
work in testing for non-functional properties.
The ASSYST system uses a simple form of tracing for
counting execution steps to gather performance
measurements~\cite{assyst}.
The Scheme-robo system counts the number of evaluations done, which can be used
for very coarse complexity analysis.
The authors conclude that better error messages are the most important area of
improvement~\cite{Saikkonen2001}.

The AI community has built tutors that aim at
bug localization by comparing source code of the student and
the teacher's programs.
LAURA~\cite{laura} converts teacher's and student's program into a
graph based representation
and compares them heuristically by applying program transformations
while reporting mismatches as potential bugs. TALUS~\cite{talus}
matches a student's attempt with a collection of teacher's algorithms.
It first tries to recognize the algorithm used and then tentatively
replaces the top-level expressions in the student's attempt
with the recognized algorithm for generating correction feedback.
In contrast, we perform trace comparison (instead of source code
comparison), which provides robustness to syntactic variations.

Striewe and Goedicke have proposed localizing bugs by trace
comparisons. They suggested creating full traces of program behavior while running test cases to
make the program behavior visible to students~\cite{striewe11}. They
have also suggested automatically comparing the student's trace to
that of a sample solution~\cite{striewe13} for generating more directed
feedback. However, no implementation has been reported. We also compare the
student's trace with the teacher's trace, but we look for
similarities as opposed to differences.

Recently it was shown that automated techniques can also
provide repair based feedback for functional correctness.
Singh's SAT solving based technology~\cite{singh:pldi13} can
successfully generate feedback (of up to 4 corrections) on around 64\%
of all incorrect solutions (from an MIT introductory programming
course) in about 10 seconds on average. While test inputs provide
guidance on {\em why} a given solution is incorrect and bug
localization techniques provide guidance on {\em where} the error
might be, repairs provide guidance on {\em how} to fix an incorrect
solution. We also provide repair suggestions that are manually associated with
the various teacher specifications, but for performance based
aspects. 
Furthermore, our suggestions are not necessarily restricted to small fixes.

\subsection{Performance Analysis}
The Programming Languages and Software Engineering communities have
explored various kinds of techniques to generate performance
related feedback for programs.
Symbolic execution based techniques
have been used for identifying non-termination related
issues~\cite{gupta:popl08,burnim:ase09}.
The SPEED project investigated use of static analysis techniques for
estimating symbolic computational complexity of
programs~\cite{conf/sas/ZulegerGSV11,speed:popl09,speed:pldi10}. Goldsmith et.al.~used dynamic
analysis techniques for empirical computational
complexity~\cite{goldsmith07}.
The Toddler tool reports a specific pattern: computations with repetitive and similar
memory-access patterns~\cite{Nistor2013}.
The Cachetor tool reports memoization opportunities by identifying operations that
generate identical values~\cite{Nguyen2013}.
In contrast, we are interested in not
only identifying whether or not there is a performance issue, but also identifying its
root cause and generating repair suggestions.

\bibliographystyle{abbrv}
\bibliography{fse}

\end{document}